\documentclass{article}
\usepackage{amsmath}
\usepackage{amsfonts}
\usepackage{amsthm}
\allowdisplaybreaks
\usepackage[colorlinks=true]{hyperref}

\DeclareFontFamily{OT1}{pzc}{}
\DeclareFontShape{OT1}{pzc}{m}{it}{<-> s * [1.10] pzcmi7t}{}
\DeclareMathAlphabet{\mathpzc}{OT1}{pzc}{m}{it}

\newcommand{\sito}{{\sc sito}}
\newcommand{\sitos}{{\sc sito\rm s}}
\newcommand{\cito}{{\sc cito}}
\newcommand{\citos}{{\sc cito\rm s}}
\newcommand{\ver}[1]{\widehat{#1}}
\newcommand{\verr}[1]{\ver{r}_{#1}}
\newcommand{\verrp}[1]{\ver{r}_{#1}\,'}
\newcommand{\verp}[1]{\ver{p}_{#1}}
\newcommand{\verq}[1]{\ver{q}_{#1}}
\newcommand{\verqp}[1]{\ver{q}_{#1}\,'}

\newcommand{\veps}[1]{\ver{\varepsilon}_{#1}}

\newcommand{\eps}{\varepsilon}
\newcommand{\CG}[6]{\langle #1,#2;#3,#4 | #5,#6\rangle}
\renewcommand{\O}{{\Omega}}
\renewcommand{\H}{{\cal H}}
\newcommand{\R}{{\cal R}}
\renewcommand{\P}{{\cal P}}

\newcommand{\vj}{\vec{\jmath}}
\newcommand{\mmp}[3]{m^{#1 1}_{#2#3}}
\theoremstyle{definition}

\newtheorem*{thm*}{Theorem}
\newtheorem{lem}{Lemma}[section]
\newtheorem{crl}{Corollary}[section]
\newtheorem*{crl*}{Corollary}
\title{The Wigner-Eckart Theorem for Reducible Symmetric Cartesian
  Tensor Operators}
\author{Antonio O.\ Bouzas \thanks{E-mail:
    abouzas@mda.cinvestav.mx}\\\small Departamento de F\'{\i}sica
  Aplicada, CINVESTAV-IPN \\\small Carretera Antigua a Progreso Km.\
  6, Apdo.\ Postal 73 ``Cordemex''\\\small M\'erida 97310, Yucat\'an,
  M\'exico}
\date{\today}
\begin{document}

\maketitle

\begin{abstract}
  \noindent We explicitly establish a unitary correspondence between
  spherical irreducible tensor operators and cartesian tensor
  operators of any rank.  That unitary relation is implemented by
  means of a basis of integer-spin wave functions that constitute
  simultaneously a basis of the spaces of cartesian and spherical
  irreducible tensors.  As a consequence, we extend the Wigner--Eckart
  theorem to cartesian irreducible tensor operators of any rank, and
  to totally symmetric reducible ones. We also discuss the tensorial
  structure of several standard spherical irreducible tensors such as
  ordinary, bipolar and tensor spherical harmonics, spin-polarization
  operators and multipole operators. As an application, we obtain an
  explicit expression for the derivatives of any order of spherical
  harmonics in terms of tensor spherical harmonics.\\[3pt]
\end{abstract}

\newpage
\tableofcontents
\newpage

\section{Introduction}
\label{sec:intro}

The Wigner--Eckart theorem is one of the fundamental results in
quantum angular-momentum theory. As is well known, it states that the
dependence on magnetic quantum numbers of the matrix elements of
spherical irreducible tensor operators (henceforth \sitos) between
angular--momentum eigenstates, is factorizable into a Clebsch--Gordan
(henceforth CG) coefficient. This leads to a drastic simplification of
the calculation of tensor-operator matrix elements, such as those
appearing in perturbative computations in molecular, atomic and
nuclear systems (and more generally in rotationally invariant
many-body problems), and a vast array of other quantum-mechanical
systems such as, \emph{e.g.}, the theory of anisotropic liquids
\cite{kro15}.  By now standard textbook material, the Wigner--Eckart
theorem was first formulated by Eckart \cite{eck30} for rank-1 \sitos\
and generalized by Wigner \cite{wig59} to \sitos\ of any rank.
Wigner's treatment is based on group--theoretic methods involving
finite rotation operators.  The definition of \sitos\ and the proof of
the Wigner--Eckart theorem based on angular-momentum commutation
relations (i.e., on infinitesimal rotation operators), as usually
found in textbooks \cite{lau77,mes61,coh77,gal90}, is due to Racah
\cite{rac42}.\footnote{A more detailed historical account is given in
  \cite{bie09}.}

In this paper we consider the relationship between \sitos\ and
cartesian irreducible tensor operators (henceforth \citos).  We
explicitly establish a unitary correspondence between them valid for
any rank.  The precise relation between spherical and cartesian
tensors allows us to apply the technical machinery of tensor algebra
to spherical tensors and, conversely, the techniques of quantum
angular--momentum theory to cartesian tensors.  That interplay is, in
fact, the main subject of this paper.  In this respect, our results
are a significant extension of the classic work of Zemach \cite{zem65}
and complementary to more recent results (\emph{e.g.}, \cite{nor82}
and references cited there).

Another important feature of our approach is that the coefficients of
the unitary transformation relating \sitos\ and \citos\ have a
well-defined physical meaning: they are an orthonormal, complete set
of standard spin wave--functions, satisfying the eigenvalue equations,
phase conventions and complex--conjugation properties expected of
angular--momentum eigenstates and eigenfunctions which, furthermore,
form a Clebsch--Gordan series of angular--momentum states. They
transform under rotations either as cartesian or as spherical tensors,
which explains their role in relating both types of tensors.  By means
of the relation between spherical and cartesian irreducible tensor
operators we extend the Wigner--Eckart theorem to \citos\ of any rank.
Remarkably, such an extension has not been considered before in the
literature.  We discuss also the cartesian tensorial form of several
commonly occurring \sitos\ (ordinary spherical harmonics, as well as
bipolar and tensor ones, among others).  Those cartesian expressions
provide a viewpoint complementary to the usual analytical one, by
making the tensorial structure of \sitos\ completely explicit, which
makes possible to obtain relations that would otherwise be more
difficult to find.  Furthermore, writing \sitos\ in tensorial form is
of interest in the context of relativistic theories, for example in
connection with covariant partial--wave expansions.  The converse case
is also true, since by mapping cartesian tensors into spherical ones
their angular--momentum properties become apparent.

Reducible cartesian tensors occur frequently in physics, so the
evaluation of the angular-momentum matrix elements of reducible tensor
operators is clearly of interest.  We obtain in this paper a further
extension of the Wigner--Eckart theorem to a limited class of
reducible tensor operators of any rank, namely, that of totally
symmetric ones.  This allows us to compute the matrix elements of
tensor powers of the position and of the momentum operators.  As an
application, we obtain an explicit expression for the gradients of
spherical harmonics to all orders in terms of tensor spherical
harmonics.  That result is a generalization of the well-known gradient
formula \cite{bet33,gal90,var88} for first derivatives to derivatives
of any order.

The paper is organized as follows.  In the following section we
introduce the spin operator for cartesian tensors and briefly discuss
our notation and conventions for tensors. In section \ref{sec:spin.wf}
we construct a standard basis of spin wave--functions for any integer
spin and establish their main properties both as angular--momentum
eigenfunctions and as a basis of the space of cartesian irreducible
tensors.  By means of that basis, in section \ref{sec:csitos} we
obtain the unitary relation between \sitos\ and \citos.  In section
\ref{sec:we} we establish the Wigner--Eckart theorem for \citos. In
section \ref{sec:standard} we analyze several \sitos\ commonly used in
the literature, including ordinary, bipolar, and tensor spherical
harmonics, spin--polarization operators, and electric multipole
operators, from the point of view of their relation to \citos.  The
Wigner--Eckart theorem for totally symmetric reducible tensors is
discussed in section \ref{sec:totsym}.  Its application to the
computation of derivatives of spherical harmonics to any order is
worked out in section \ref{sec:dershy}.  In section \ref{sec:parirr}
we discuss partially irreducible cartesian tensors, provide the
extension of the Wigner-Eckart theorem to them, and discuss the magnetic
multipole expansion.  Finally, in section \ref{sec:finrem} we give
some concluding remarks.  Our conventions for angular momentum theory
are detailed in appendix \ref{sec:angmom}.

\section{The spin operator for cartesian tensors.}
\label{sec:cartens}

Throughout this paper, unless otherwise indicated, we consider only
proper rotations represented by proper orthogonal matrices $\R$ with
$\det\R=1$, so we will not need to distinguish between tensors and
pseudo-tensors.  Given a fixed orthogonal coordinate frame, we denote
coordinate versors either by $\ver{x}$, $\ver{y}$, $\ver{z}$ or by
$\ver{e}\,^i$ ($i=1$, 2, 3), with components
$\ver{e}\,^i_j=\delta_{ij}$ so that $\ver{e}^{1,2,3}=\ver{x}$,
$\ver{y}$, $\ver{z}$.  An infinitesimal rotation is of the form
$\R_{ij}(\delta\vec{\theta}\,)=\delta_{ij}+\delta\R_{ij}$, with
$\delta\R_{ij}=\eps_{ikj} \delta\theta_k$ specified by the
infinitesimal parameter $\delta\vec{\theta}$ describing a rotation by
an infinitesimal angle $|\delta\vec\theta|$ about the axis
$\ver{\theta}=\delta\vec\theta/|\delta\vec\theta|$.  By definition, a
rank-$n$ cartesian tensor transforms under an infinitesimal rotation
$\R(\delta\vec{\theta})$ as
\begin{equation}
\label{eq:tensor.draft}
\delta A_{i_1\ldots i_n}= \delta\theta_j (\eps_{i_1jk} A_{ki_2\ldots i_n}+
\eps_{i_2jk} A_{i_1ki_3\ldots i_n}+  \ldots +
\eps_{i_njk} A_{i_1\ldots i_{n-1}k}).  
\end{equation}
It is convenient to state the definition (\ref{eq:tensor.draft}) in a 
more compact way by gathering the coefficients on its
right-hand side in a linear operator
$\vec{S}_{(n)}$  
\begin{equation}
  \label{eq:c.tensors}  
\delta A_{i_1\ldots i_n}= -i\delta\theta_j (S_{(n)j})_{i_1\ldots
  i_n;k_1\ldots k_n}  A_{k_1\ldots k_n},
\qquad
  \left(S_{(n)j}\right)_{i_1\ldots i_n;k_1\ldots k_n} = 
\sum_{t=1}^n i \veps{i_tjk_t}\prod_{\substack{p=1\\p\neq t}}^n
\delta_{i_pk_p}. 
\end{equation}
In particular $(S_{(1)j})_{i;k}=i\eps_{ijk}$ and $\R(\vec{\theta}) =
\exp(-i\vec{\theta}\cdot\vec{S}_{(1)})$ .  From (\ref{eq:c.tensors})
it follows that the linear operators $\vec{S}_{(n)}$ are hermitian
and, as expected, satisfy angular-momentum commutation relations
\begin{equation}
  \label{eq:spin.commutator}
[S_{(n)k}, S_{(n)h}]= i \eps_{khr}S_{(n)r}.  
\end{equation}
We will therefore refer to $\vec{S}_{(n)}$ as the ``spin matrices'' in
the space of rank-$n$ cartesian tensors.  Another important
consequence of the definition (\ref{eq:c.tensors}) is the recursion
relation 
\begin{multline}
  \label{eq:s.mat.prop.a}
(S_{(p+q)k})_{i_1\ldots i_{p+q};j_1\ldots j_{p+q}} =
(S_{(p)k})_{i_1\ldots i_p;j_1\ldots j_p} \delta_{i_{p+1}j_{p+1}} 
\ldots \delta_{i_{p+q}j_{p+q}} \\
+ \delta_{i_{1}j_{1}}\ldots \delta_{i_{p}j_{p}} 
(S_{(q)k})_{i_{p+1}\ldots i_{p+q};j_{p+1}\ldots j_{p+q}},
  \end{multline}
of which (\ref{eq:c.tensors}) itself is the solution.  From
(\ref{eq:s.mat.prop.a}) we obtain an equivalent relation for finite
rotation operators,
\begin{equation}
  \label{eq:r.mat.prop}
\left(e^{-i \vec{\theta}\cdot\vec{S}_{(p+q)}}\right)_{i_1\ldots
    i_{p+q};j_1\ldots j_{p+q}} =
\left(e^{-i \vec{\theta}\cdot\vec{S}_{(p)}}\right)_{i_1\ldots
    i_p;j_1\ldots j_p}
\left(e^{-i \vec{\theta}\cdot\vec{S}_{(q)}}\right)_{i_{p+1}\ldots
    i_{p+q};j_{p+1}\ldots j_{p+q}}.
\end{equation}
This equality is consistent with the fact that a rank $n$ complex
cartesian tensor transforms under rotations as the tensor product of
$n$ vectors.  Indeed, the r.h.s.\ of (\ref{eq:r.mat.prop}) is the
rotation matrix in the space of rank-$n$ tensors, and through
iteration of the equality (\ref{eq:r.mat.prop}) we get
\begin{equation}
  \label{eq:finite.rot.tens}
  \left(e^{-i \vec{\theta}\cdot\vec{S}_{(n)}}\right)_{i_1\ldots
    i_n;j_1\ldots j_n} = 
\left(e^{-i \vec{\theta}\cdot\vec{S}_{(1)}}\right)_{i_1;j_1}
\ldots \left(e^{-i
    \vec{\theta}\cdot\vec{S}_{(1)}}\right)_{i_{n};j_{n}} = 
\R(\vec{\theta}\,)_{i_1j_1}\ldots
  \R(\vec{\theta}\,)_{i_nj_n}. 
\end{equation}
A direct proof of (\ref{eq:r.mat.prop}) or (\ref{eq:finite.rot.tens})
can be obtained by differentiating the equality with respect to
$\theta=|\vec{\theta}|$, with $\ver{\theta}$ fixed, to find that both
sides satisfy the same first--order differential equation.
Alternatively, one can expand the exponentials on both sides of
(\ref{eq:r.mat.prop}) or (\ref{eq:finite.rot.tens}) in powers of
$\theta$ and use (\ref{eq:s.mat.prop.a}) to obtain a binomial
expansion for the powers in each term.  Both procedures are
straightforward, though somewhat tedious, so we omit the details for
brevity. We include here for later reference the relation
\begin{equation}
  \label{eq:s.mat.prop.b}
(\vec{S}_{(n)}^2)_{i_1\ldots i_n;h_1\ldots h_n} = 2n \prod_{q=1}^n
\delta_{i_qh_q} +\sum^n_{\substack{p,t=1\\p\neq t}} (\delta_{i_ph_t}
\delta_{i_th_p}- \delta_{i_pi_t} \delta_{h_ph_t})
\prod^n_{\substack{q=1\\p\neq q\neq t}}\delta_{i_qh_q},
\end{equation}
which gives the expression for the matrix of the squared spin operator.

The set of all rank-$n$ complex irreducible (i.e., totally symmetric
and traceless \cite{ham62}) tensors is a $2n+1$-dimensional linear
subspace of the space of complex rank-$n$ tensors.  That subspace is
invariant under $\vec{v}\cdot\vec{S}_{(n)}$ for any vector $\vec{v}$
since, as can be seen from (\ref{eq:c.tensors}),
$v_k(S_{(n)k})_{i_1\ldots i_n;j_1\ldots j_n} A_{j_1\ldots j_n}$ and
$v_k(S_{(n)k})_{i_1\ldots i_n;j_1\ldots j_n} A_{i_1\ldots i_n}$ are
totally symmetric and traceless in their free indices if the tensor
$A_{h_1\ldots h_n}$ is irreducible.  Similarly, $\R_{i_1
  j_1}\ldots\R_{i_n j_n}A_{j_1\ldots j_n}$ is irreducible if
$A_{i_1\ldots i_n}$ is.  Given a tensor $A_{i_1\ldots i_n}$ we define
its associated totally symmetrized tensor as
\begin{equation}
  \label{eq:symmetrized}
  A_{\{i_1\ldots i_n\}} = \sum_\sigma A_{i_{\sigma_1}\ldots i_{\sigma_n}},
\end{equation}
where the sum extends over all permutations $i_{\sigma_1}\ldots
i_{\sigma_n}$ of $i_1\ldots i_n$.  The totally symmetric part of
$A_{i_1\ldots i_n}$ is then $1/n! A_{\{i_1\ldots i_n\}}$.  Similarly,
we denote by $A_{(i_1\ldots i_n)_0}$ the traceless part of
$A_{i_1\ldots i_n}$ (for example, $r_{(i}r_{j)_0}=r_i r_j-1/3
|\vec{r}\,|^2 \delta_{ij}$).  The traceless part of the totally
symmetrized tensor associated to $A_{i_1\ldots i_n}$ is then denoted
$A_{\{i_1\ldots i_n\}_0}$.  The irreducible component of $A_{i_1\ldots
  i_n}$ is therefore $1/n!A_{\{i_1\ldots i_n\}_0}$.  We provide a
practical method to compute the irreducible part of any tensor in the
following section (see equation (\ref{eq:projector.applied}) below).

\section{A basis for irreducible tensors and integer--spin wave--functions}
\label{sec:spin.wf}

In this section we construct an orthonormal basis for the
$(2s+1)$-dimensional space of irreducible rank-$s$ tensors, $s\geq 0$
integer.  The basis irreducible tensors are eigenfunctions of
$\vec{S}\,^2$ and $\ver{z}\cdot\vec{S}$, and satisfy the phase
conventions required of standard angularm--momentum eigenfunctions
(see appendix \ref{sec:angmom}).  They are, therefore, also a basis of
spin-$s$ wave--functions.  Furthermore, as shown in the following
section, those basis tensors are the matrix elements of the unitary
transformation mapping spherical irreducible tensor operators into
cartesian ones.

The basis of spin-1 wave functions consists of the simultaneous
eigenvectors of $(\vec{S}_{(1)}^2)_{i;j}=2\delta_{ij}$ and
$(\ver{z}\cdot \vec{S}_{(1)})_{i;j} = i \eps_{i3j}$.  We choose their
global phase so as to obtain the usual polarization unit vectors
\begin{subequations}
  \label{eq:stdvec}
\begin{equation}
  \label{eq:stdvec.a}
  \veps{(1)}(\pm1) = \mp\frac{1}{\sqrt{2}}(\ver{x}\pm i \ver{y}),
\qquad
  \veps{(1)}(0) = \ver{z}.
\end{equation}
We can also write, more compactly,
\begin{equation}
  \label{eq:stdvec.b}
  \veps{(1)i}(m) = \sqrt{\frac{4\pi}{3}}Y_{1m}(\ver{e}\,^i),
\qquad
   m=0,\pm1,
\end{equation}
\end{subequations}
with $Y_{1m}$ a spherical harmonic, an equality that can easily be
checked and whose origins are explained below in section
\ref{sec:standard}.  The basis vectors (\ref{eq:stdvec}) possess the
following orthonormality, complex conjugation, and completeness
properties
\begin{equation}
  \label{eq:stdvec.prop}
  \veps{(1)}(m')^*\cdot\veps{(1)}(m)=\delta_{m'm},
\quad
\veps{(1)}(m)^*=(-1)^m \veps{(1)}(-m),
\quad
\sum_{m=-1}^1 \veps{(1)i}(m)\veps{(1)j}(m)^* = \delta_{ij}.
\end{equation}
From (\ref{eq:stdvec.a}) and (\ref{eq:cartj}) we find that 
\begin{equation}
  \label{eq:spin.mat.elm.1}
  \langle 1, m' | S_k | 1, m\rangle =
  \veps{(1)i}(m')^* (S_{(1)k})_{i;j} \veps{(1)j}(m),
\end{equation}
so the vectors (\ref{eq:stdvec}) do satisfy the standard conventions
for angular-momentum wave-functions, and in particular the
Condon--Shortley phase convention \cite{con57,wig59,edm96,gal90} (see
Appendix \ref{sec:angmom}).

Before discussing the general case of rank-$s$ tensors it is
convenient to briefly consider first the case $s=2$.  Thus, we look
for a basis of spin-2 wave functions consisting of rank-2 tensors
$\veps{(2)ij}(m)$, $-2\leq m\leq2$.  Because those basis
wave-functions must be eigenfunctions of $\vec{S}^2$ with quantum
number $s=2$, without admixture of states with other $s$, the tensor
$\veps{(2)ij}(m)$ must be irreducible.  Otherwise, some non-vanishing
linear combination of its components would exist that transforms as a
lower-rank tensor, therefore representing states of spin 1 or 0.
Furthermore, the tensors $\veps{(2)ij}(m)$ must be orthonormal and
satisfy the usual conventions (\ref{eq:ladderj}) for angular-momentum
eigenstates.  Given our spin-1 wave functions (\ref{eq:stdvec}), the
general theory of angular momentum indicates that the sought--for
rank-2 tensors are given by
  \label{eq:stdtens}
\begin{equation}
  \label{eq:stdtens.a}
\veps{(2)ij}(m) = \sum_{m_1,m_2=-1}^1 \CG{1}{m_1}{1}{m_2}{2}{m} 
\veps{(1)i}(m_1) \veps{(1)j}(m_2).
\end{equation}
Explicit evaluation of the r.h.s.\ of this equation shows that
$\veps{(2)}(m)$ are symmetric and traceless, therefore irreducible.
The spin operator $\vec{S}$ is represented in the space of rank-2
tensors by the spin matrix $\vec{S}_{(2)}$ from
(\ref{eq:s.mat.prop.a}), with 
\begin{equation}
  \label{eq:squared.spin.matrix.2}
  (\vec{S}_{(2)}^2)_{i_1i_2;k_1k_2} =
  4 \delta_{i_1k_1}\delta_{i_2k_2} - 2 \delta_{i_1i_2}\delta_{k_1k_2}
  + 2 \delta_{i_1k_2} \delta_{i_2k_1}.
\end{equation}
We see from (\ref{eq:squared.spin.matrix.2})
that irreducible tensors are eigenstates of $\vec{S}_{(2)}^2$ with
eigenvalue $s(s+1)=6$, or $s=2$, antisymmetric tensors are eigenstates
with $s=1$ and tensors that are multiples of the identity correspond
to $s=0$.  This decomposition of rank-2 tensor space corresponds, of
course, to the usual decomposition into traceless symmetric,
antisymmetric, and trace parts.

We now turn to the general case $s\geq2$.  The basis spin
wave-functions must be $2s+1$ rank-$s$ irreducible tensors
$\veps{(s)i_1\ldots i_s}(m)$, $-s\leq m\leq s$, constituting an
orthonormal set satisfying the conventions (\ref{eq:cartj}) for
angular-momentum states.  Furthermore, as functions of the spin $s$,
they should be members of a Clebsch--Gordan series of angular-momentum
states.  Having already found the basis tensors for $s=1$, 2, we
proceed recursively to define
\begin{equation}
  \label{eq:std.tens.s}
\veps{(s)i_1\ldots i_s}(m) = \sum_{m_1=-s+1}^{s-1} \sum_{m_2=-1}^1 
\CG{s-1}{m_1}{1}{m_2}{s}{m} \veps{(s-1)i_1\ldots i_{s-1}}(m_1) 
\veps{(1)i_s}(m_2),
\quad
-s\leq m\leq s.
\end{equation}
From this definition and the standard properties of CG
coefficients \cite{gal90,var88}, we can easily derive the
orthonormality and complex-conjugation relations
\begin{equation}
  \label{eq:std.tens.s.prop}
\veps{(n)i_1\ldots i_n}(m')^* \veps{(n)i_1\ldots i_n}(m) =
\delta_{m'm},  
\qquad
\veps{(n)i_1\ldots i_n}(m)^* = (-1)^m \veps{(n)i_1\ldots i_n}(-m).
\end{equation}
From the first equality we see that the tensors (\ref{eq:std.tens.s})
do form an orthonormal set which is, therefore, a basis of a
$(2s+1)$-dimensional subspace of the space of rank-$s$ complex
tensors.  In order to identify that subspace with the subspace of
irreducible tensors, which has the same dimension, we have to prove
that the basis tensors are irreducible.  For that purpose, we notice
that (\ref{eq:std.tens.s}) is a recursion relation with known
coefficients and initial condition (\ref{eq:stdvec}).  Exploiting the
fact that the explicit expression for CG coefficients
\cite{gal90,var88} coupling angular momenta that differ by one unit,
as in (\ref{eq:std.tens.s}), is rather simple, we can solve the
recursion by iteration to find 
\begin{equation}
  \label{eq:std.tens.s.alt}
  \veps{(n)i_1\ldots i_n}(m) = 
\left(\frac{(n+m)!(n-m)!}{(2n)!}\right)^\frac{1}{2} 
\sum_{\substack{s_1,\ldots,s_n=-1\\
      s_1+\cdots+s_n=m}}^1
(\sqrt{2})^{n-\sum_{h=1}^n |s_h|}\,
\veps{(1)i_1}(s_1)\ldots\veps{(1)i_n}(s_n).   
\end{equation}
This expression provides an explicit definition of $\veps{i_1\ldots
  i_n}$, equivalent to (\ref{eq:std.tens.s}). It also shows that
$\veps{(n)i_1\ldots i_n}$ is totally symmetric. Thus, in order to
prove that it is also totally traceless it is enough to show that it
is traceless with respect to the first pair of indices.  That follows
by induction, since $\veps{(2)jj}(m)=0$ as follows by explicit
computation, and since $\veps{(s-1)jji_3\ldots i_{(s-1)}}(m)=0$
implies $\veps{(s)jji_3\ldots i_{s}}(m)=0$, by (\ref{eq:std.tens.s}).
Besides the recursive and explicit definitions (\ref{eq:std.tens.s})
and (\ref{eq:std.tens.s.alt}), an implicit definition of $\veps{(s)}$
can  also be given
\begin{equation}
  \label{eq:std.tens.s.impl}
  \veps{(s)i_1\ldots i_s}(m) = \sqrt{\frac{4\pi}{s!(2s+1)!!}}
  \partial_{i_1} \ldots \partial_{i_s}\left( |\vec{r}|^s Y_{s
      m}(\verr{})\right).
\end{equation}
This equality will be proved below, in section \ref{sec:shy}.  The
total symmetry of $\veps{(\ell)i_1\ldots i_\ell}(m)$ is apparent in
(\ref{eq:std.tens.s.impl}), and its tracelessness follows because
$|\vec{r}|^\ell Y_{\ell m}(\verr{})$ is a solution to the Laplace equation.

Since the set of $(2s+1)$ spin-$s$ wave functions
(\ref{eq:std.tens.s}) is an orthonormal basis of the subspace of
irreducible tensors, the orthogonal projector from the space of
rank-$s$ tensors onto that subspace must be given by
\begin{equation}
  \label{eq:X.projector}
X_{i_1\ldots i_s;j_1\ldots j_s} = \sum_{m=-s}^s \veps{(s)i_1\ldots
  i_s}(m)   \veps{(s)j_1\ldots j_s}(m)^*  =
\sum_{m=-s}^s \veps{(s)i_1\ldots i_s}(m)^*\veps{(s)j_1\ldots j_s}(m).  
\end{equation}
If the left-hand side of this equality is computed explicitly,
(\ref{eq:X.projector}) constitutes a completeness relation for the
standard tensors (\ref{eq:std.tens.s}).  In \cite{bou1} an explicit
expression is given for $s=2$, 3, as well as an algebraic expression
valid for any $s$.  Those expressions are not particularly useful for
the purposes of this paper, however, so we omit them for brevity.
Rather, we shall regard (\ref{eq:X.projector}) as an explicit
expression for the projector $X_{i_1\ldots i_s;j_1\ldots j_s}$.  Its
usefulness is illustrated below in (\ref{eq:projector.applied}).

From (\ref{eq:std.tens.s.alt}) we find the two simple relations
\begin{equation}
  \label{eq:std.tens.upper.lower}
  \veps{(s)i_1\ldots i_s}(\pm s) = \veps{(1)i_1}(\pm1)\dots
  \veps{(1)i_s}(\pm1). 
\end{equation}
These equalities are useful, together with standard recoupling
identities, to compute reduced matrix elements.  Furthermore, they
imply $\veps{(n_1+n_2)i_1\ldots i_{n_1+n_2}}(\pm(n_1+n_2)) =
\veps{(n_1)i_1\ldots i_{n_1}}(\pm n_1)\veps{(n_2)i_{n_1+1}\ldots
  i_{n_1+n_2}}(\pm n_2)$, so the basis spin wave-functions
(\ref{eq:std.tens.s}) comply also with the Condon--Shortley phase
convention for coupled angular momentum states \cite{con57,edm96},
$|j_1,j_2,j_1+j_2,j_1+j_2\rangle = |j_1,j_1;j_2,j_2\rangle$.  Another
important property of the basis tensors (\ref{eq:std.tens.s}) is the
equality
\begin{equation}
  \label{eq:maximal.coupling}
  \veps{(n_1+n_2)k_1\ldots k_{n_1}h_1\ldots h_{n_2}}(m) =
  \sum_{m_1=-n_1}^{n_1} \sum_{m_2=-n_2}^{n_2} 
  \CG{n_1}{m_1}{n_2}{m_2}{n_1+n_2}{m} \veps{(n_1)k_1\ldots k_{n_1}}(m_1)  
  \veps{(n_2)h_{1}\ldots h_{n_2}}(m_2),
\end{equation}
which shows that the maximal coupling of two standard tensors is again
a standard tensor, and of which (\ref{eq:std.tens.s}) is the
particular case $n_1=1$ or $n_2=1$. It is possible to derive
(\ref{eq:maximal.coupling}) directly by substituting
(\ref{eq:std.tens.s}) on its right-hand side and using recoupling
identities \cite{bou1}.  Here we give a less direct proof.  Both
sides of the equality (\ref{eq:maximal.coupling}) are by construction
eigenstates of $\vec{S}^2$ and $\ver{z}\cdot\vec{S}$ with the same
eigenvalues. Thus, for $m<n_1+n_2$ both sides are obtained by repeated
application of $S_-$ to their $m=n_1+n_2$ values. Since for
$m=n_1+n_2$ both sides are seen to be equal by
(\ref{eq:std.tens.upper.lower}), the equality holds for $m<n_1+n_2$ as
well.

The spin operator $\vec{S}$ is represented in the space of rank-$n$
tensors by the spin matrix (\ref{eq:c.tensors}).  Notice that
(\ref{eq:s.mat.prop.a}) means that $\vec{S}_{(n)} = \vec{S}_{(n-1)}
\otimes I_{1} + I_{n-1}\otimes \vec{S}_{(1)}$, which is consistent
with the inductive definition (\ref{eq:std.tens.s}).
From (\ref{eq:std.tens.s}) and
(\ref{eq:c.tensors}) we can show the fundamental relation 
\begin{equation}
  \label{eq:spin.matrix.element}
  \langle n,m' |S_k| n,m \rangle = \veps{(n)i_1\ldots i_n}(m')^*
 (S_{(n)k})_{i_1\ldots i_n;j_1\ldots j_n} \veps{(n)j_1\ldots j_n}(m), 
\end{equation}
where the left-hand side is a standard angular-momentum matrix element
as given by (\ref{eq:cartj}).  A detailed proof of
(\ref{eq:spin.matrix.element}) is given at the end of appendix
\ref{sec:angmom}.  Multiplying both sides of
(\ref{eq:spin.matrix.element}) by $\veps{(n)}(m')$ and
summing over $m'$ we derive the equivalent relation
\begin{equation}
  \label{eq:spin.matrix.action}
  \begin{aligned}
    (S_{(n)k})_{i_1\ldots i_n;j_1\ldots j_n} \veps{(n)j_1\ldots
      j_n}(m) &= \sum_{m'=-n}^n \veps{(n)i_1\ldots i_n}(m')
    \veps{(n)h_1\ldots h_n}(m')^*
    (S_{(n)k})_{h_1\ldots h_n;j_1\ldots j_n} \veps{(n)j_1\ldots j_n}(m)\\
&=\sum_{m'=-n}^n \veps{(n)i_1\ldots i_n}(m')    
  \langle n,m' |S_k| n,m \rangle,
  \end{aligned}
\end{equation}
where in the first equality we used the fact that the tensor on the
left-hand side is irreducible, therefore invariant under the projector
(\ref{eq:X.projector}).  From (\ref{eq:spin.matrix.element}) and
(\ref{eq:spin.matrix.action}) we can inductively prove a
generalization of (\ref{eq:spin.matrix.element}) to any number of
spin-operator components
\begin{equation}
  \label{eq:spin.matrix.element.multi}
\langle n,m' |S_{k_1}\ldots S_{k_p}| n,m \rangle = \veps{(n)i_1\ldots 
  i_n}(m')^* 
(S_{(n)k_1}\ldots S_{(n)k_p})_{i_1\ldots i_n;j_1\ldots j_n}
\veps{(n)j_1\ldots j_n}(m),
\quad 
p\geq1, 
\end{equation}
and from this relation the corresponding generalization of
(\ref{eq:spin.matrix.action}) follows.
The squared spin operator is
represented by (\ref{eq:s.mat.prop.b}).  It is not difficult to verify
from that equation that any rank-$n$ irreducible tensor is an
eigenfunction of $\vec{S}_{(n)}^2$ with eigenvalue $n(n+1)$. Lower
eigenvalues correspond to tensors with less symmetry or with
non-vanishing traces.

Lastly, we notice that, since the standard tensors (\ref{eq:std.tens.s})
or their complex conjugates constitute an orthonormal basis of the
linear space of rank-$s$ irreducible tensors, given any rank-$n$ complex tensor
$A_{i_1\ldots i_n}$ its irreducible part can be written as
\begin{equation}
  \label{eq:projector.applied}
  \frac{1}{n!} A_{\{i_1\ldots i_n\}_0} =
\sum_{m=-n}^n \veps{(n)i_1\ldots i_n}(m)^* \veps{(n)j_1\ldots j_n}(m) 
A_{j_1\ldots j_n}.
\end{equation}
In fact, (\ref{eq:projector.applied}) provides a practical method to
obtain the irreducible part of a reducible tensor.  If the tensor
$A_{i_1\ldots i_n}$ under consideration is irreducible, then the
left-hand side of (\ref{eq:projector.applied}) is equal to
$A_{i_1\ldots i_n}$.

\subsection{Finite rotations}
\label{sec:finrot}

We turn next to the transformation properties of $\veps{(s)}$ under
finite rotations.  The theory of finite rotations in quantum mechanics
is well known 
(see \cite{wig59,lau77,mes61,coh77,gal90,bie09,var88,edm96}).
The generator of infinitesimal rotations is the total
angular--momentum operator $\vec{J}$.  Thus, in terms of the normal
parameters $\vec{\theta}$ the unitary rotation operator is given by
$U(\R(\vec{\theta}\,)) = \exp(-i\vec{\theta}\cdot\vec{J})$.
Similarly, in terms of Euler angles we have
$U(\R(\alpha,\beta,\gamma)) = \exp(-i\alpha J_3) \exp(-i\beta J_2)
\exp(-i\gamma J_3)$.  The matrix representation of the rotation
operator $U(\R)$ in the eigenspace of total angular momentum $j$ is
defined as 
\begin{equation}
  \label{eq:finite.rot.j}
  U(\R)|j,m\rangle = \sum_{m'} |j,m'\rangle  D^j_{m'm}(\R),
\qquad
D^j_{m'm}(\R) = \langle j,m' | U(\R) | j, m\rangle. 
\end{equation}
If the rotation matrix $\R$ is parameterized as a function of the
Euler angles, the resulting rotation matrices
$D^j_{m'm}(\alpha,\beta,\gamma)$ are the Wigner $D$--matrices
\cite{wig59,bie09,var88,edm96}.  Normal parameters may also be used,
and the resulting unitary rotation matrices
$D^j_{m'm}(\vec{\theta}\,)$ (sometimes denoted
$U^j_{m'm}(\vec{\theta}\,)$ \cite{var88}) and their relation to Wigner
$D$--matrices have been extensively studied \cite{bie09,var88}.  The
generic notation $D^j_{m'm}(\R)$ used here refers to any such
parameterization.

From the definition (\ref{eq:finite.rot.j}), by using
(\ref{eq:spin.matrix.element}) and (\ref{eq:finite.rot.tens}), for
$D^j_{m'm}(\R)$ with integer $j$ we get
\begin{equation}
  \label{eq:true.D}
D^\ell_{m'm}(\R) = \veps{(\ell)h_1\ldots h_\ell}(m')^* 
\R_{h_1j_1}\ldots \R_{h_\ell j_\ell} \veps{(\ell)j_1\ldots j_\ell}(m), 
\qquad
0\leq \ell\in \mathbb{Z}, 
\end{equation}
which expresses $D^\ell_{m'm}(\R)$ as the spherical components of the
cartesian rotation matrix $(\R\otimes\cdots\otimes\R)_{h_1\ldots
  h_\ell;j_1\ldots j_\ell}$.  Substituting (\ref{eq:std.tens.s.alt})
into (\ref{eq:true.D}) leads to an expression of $D^n(\R)$ with
integer $n$ in terms of $D^1(\R)$
\begin{equation}
\label{eq:Dn1}
  \begin{aligned}
D^n_{m'm}(\R) &= \frac{2^n}{(2n)!} \sqrt{(n+m')!(n-m')!}
\sqrt{(n+m)!(n-m)!}\\
&\times
\sum_{\substack{s'_1,\ldots,s'_n=-1\\
      s'_1+\ldots+s'_n=m'}}^1
\sum_{\substack{s_1,\ldots,s_n=-1
\rule{0pt}{5.75pt}\\[2.25pt]
      s_1+\ldots+s_n=m}}^1
\frac{1}{(\sqrt{2})^{\sum_{h=1}^n(|s'_h|+|s_h|)}}
D^1_{s'_1s_1}(\R)\ldots D^1_{s'_ns_n}(\R),
  \end{aligned}
\end{equation}
which is formally analogous to (\ref{eq:shyfinal.alt}) and which, like
(\ref{eq:true.D}), holds for any parameterization used for $\R$. The
relations (\ref{eq:true.D}) and (\ref{eq:Dn1}) have not been given in
the previous literature.  

The transformation rules of the basis tensors $\veps{(s)}$ under
finite rotations are summarized by the equalities 
\begin{equation}
  \label{eq:finite.rot.eps}
  \begin{aligned}
\left(e^{-i \vec{\theta}\cdot\vec{S}_{(n)}}\right)_{i_1\ldots
i_n;j_1\ldots j_n} \veps{(n)j_1\ldots j_n}(m) &=
\R_{i_1j_1}(\vec{\theta}\,)\ldots \R_{i_nj_n}(\vec{\theta}\,)
\veps{(n)j_1\ldots j_n}(m) \\
&=  
\sum_{m'} \veps{(n)i_1\ldots i_n}(m')D^n_{m'm}(\R), 
  \end{aligned}
\end{equation}
where the first equality is (\ref{eq:finite.rot.tens}) and the second
one is a direct consequence of (\ref{eq:true.D}).  We see from
(\ref{eq:finite.rot.eps}) that under rotations the basis tensors
$\veps{(s)}$ transform equally well as cartesian or as spherical
tensors.  This property is the basis of the unitary relation between
spherical and cartesian irreducible tensor operators discussed in the
following section.

\section{Cartesian and spherical irreducible tensor operators} 
\label{sec:csitos}

Let $\vec{J}$ be an angular-momentum operator, and $|j,m\rangle$ the
simultaneous eigenstates of $\vec{J}\,^2$ and $\ver{z}\cdot\vec{J}$,
satisfying the standard conventions (see appendix \ref{sec:angmom}).  
An operator $O_{i_1\ldots i_n}$ is a rank-$n$ cartesian tensor
operator relative to $\vec{J}$ if it satisfies the commutation
relation
\begin{equation}
  \label{eq:ctensor}
  [J_k,O_{i_1\ldots i_n}] = -(S_{(n)k})_{i_1\ldots i_n;j_1\ldots j_n}
  O_{j_1\ldots j_n}, 
\end{equation}
with $\vec{S}_{(n)}$ defined in (\ref{eq:c.tensors}).  We say that
$O_{i_1\ldots i_n}$ is a cartesian irreducible tensor operator
(henceforth \cito) if it is totally symmetric and traceless. If
$O_{i_1\ldots i_n}$ is a generic tensor operator, its irreducible
component is $1/n!O_{\{i_1\ldots i_n\}_0}$.  An operator $O_{nm}$,
with integer $n$, $m$ ($n\geq0$, $-n\leq m\leq n$) is a spherical
irreducible tensor operator (henceforth \sito) relative to $\vec{J}$
if
\begin{subequations}
\label{eq:sito}
\begin{equation}
  \label{eq:sito.cart}
  [J_i,O_{nm}] = \sum_{m'=-n}^n \langle n,m'|J_i|n,m\rangle O_{nm'}.
\end{equation}
From this equation and (\ref{eq:spherj}) we get the equivalent
statement that $O_{nm}$ is a \sito\ if
\begin{equation}
  \label{eq:sito.spher}
  [\veps{(1)}(\epsilon)\cdot\vec{J},O_{nm}] = \sqrt{n(n+1)} 
\CG{n}{m}{1}{\epsilon}{n}{m+\epsilon} O_{n(m+\epsilon)},
\qquad
\epsilon=0,\pm1.
\end{equation}
\end{subequations}
If $O_{nm}$ is a \sito, then $O_{nm}^\dagger$ is not (unless $n=0$)
but $(-1)^m O^\dagger_{n(-m)}$ is. We call $O_{nm}$ hermitian if
$O_{nm}=(-1)^m O^\dagger_{n(-m)}$.  

It is well known \cite{lau77,gal90,var88,edm96} that if $a_i$ is a
vector operator then $A_{1m}$ with $A_{1(\pm1)}=\mp(1/\sqrt{2})(a_1\pm
ia_3)$, $A_{10}=a_3$ is a rank-1 \sito, and if $b_{ij}$ is a rank-2
cartesian tensor operator then $B_{2m}$ with
$B_{2(\pm2)}=(1/2)(b_{11}-b_{22}\pm2 i b_{12})$,
$B_{2(\pm1)}=\mp(b_{13}\pm ib_{21})$, $B_{20}=\sqrt{3/2} b_{33}$ is a
rank-2 \sito.  It is clear that $A_{1m}=\veps{(1)i}(m)a_i$ and
$B_{2m}=\veps{(2)ij}(m)b_{ij}$. The following Lemma generalizes those
relations to tensors of any rank.
\begin{lem}
  \label{lem:sito.from.cito}
Let $O_{i_1\ldots i_n}$ be a rank-$n$ cartesian tensor operator,
not necessarily irreducible, relative to the angular-momentum
operator $\vec{J}$.  Then $O_{nm}=\veps{(n)i_1\ldots i_n}(m)
O_{i_1\ldots i_n}$ is a rank-$n$ \sito\ relative to $\vec{J}$.
\end{lem}
\begin{proof}
\begin{align*}
[J_k,O_{nm}] &= \veps{(n)i_1\ldots i_n}(m) [J_k,O_{i_1\ldots i_n}] 
= \frac{1}{n!} \veps{(n)i_1\ldots i_n}(m) [J_k,O_{\{i_1\ldots
  i_n\}_0}]  \\
&= \veps{(n)i_1\ldots i_n}(m) \sum_{p=1}^n i \eps_{ki_pr} \frac{1}{n!}  
O_{\{i_1\ldots i_{p-1}ri_{p+1}\ldots i_n\}_0}
 \\ 
&= \veps{(n)i_1\ldots i_n}(m) \sum_{p=1}^n i \eps_{ki_pr} 
\sum_{m'=-n}^n \veps{(n)i_1\ldots i_{p-1}ri_{p+1}\ldots i_n}(m')^*
\veps{(n)q_1\ldots q_n}(m') O_{q_1\ldots q_n}
 \\
&= \sum_{m'=-n}^n 
\langle n,m'|S_k|n,m\rangle \veps{(n)q_1\ldots q_n}(m') O_{q_1\ldots 
  q_n} 
= \sum_{m'=-n}^n 
\langle n,m'|J_k|n,m\rangle O_{nm'},
    \end{align*}
where the second equality holds because $\veps{(n)}$ is irreducible,
the third one because $O_{\{i_1\ldots i_n\}_0}$ is a cartesian tensor
operator, the fourth one because of (\ref{eq:projector.applied}), and
the fifth one by (\ref{eq:spin.matrix.element}). In the last equality
we used the fact that $\vec{S}$ and  $\vec{J}$ are both
angular-momentum operators and, therefore, their matrix 
elements are both given by (\ref{eq:cartj}).
\end{proof}

Reciprocally, by means of the spin wave--functions of section
\ref{sec:spin.wf}, \citos\ can be obtained from \sitos. 
\begin{lem}
  \label{lem:cito.from.sito}
Let $O_{nm}$ be a rank-$n$ \sito\ relative to the
angular-momentum operator $\vec{J}$.  Then $O_{i_1\ldots i_n}=
\sum_{m=-n}^n \veps{(n)i_1\ldots i_n}(m)^* O_{nm}$ is a rank-$n$
\cito\ relative to $\vec{J}$.  
\end{lem}
\begin{proof}
\begin{align*}
[J_k,O_{i_1\ldots i_n}] &=  
\sum_{m=-n}^n \veps{(n)i_1\ldots i_n}(m)^* [J_k, O_{nm}]
= \sum_{m=-n}^n \veps{(n)i_1\ldots i_n}(m)^* 
\sum_{m'=-n}^n \langle n,m'|J_k|n,m\rangle O_{nm'}
\\
&=\sum_{m'=-n}^n O_{nm'}
\sum_{m=-n}^n \langle n,m'|S_k|n,m\rangle 
\veps{(n)i_1\ldots i_n}(m)^* 
\\
&=\sum_{m'=-n}^n O_{nm'}
\sum_{m=-n}^n \sum_{t=1}^n \veps{(n)q_1\ldots q_{t-1}rq_{t+1}\ldots 
  q_n}(m')^*
i\eps_{rkq_t} \veps{(n)q_1\ldots q_n}(m) \veps{(n)i_1\ldots i_n}(m)^* 
\\
&=\sum_{t=1}^n i\eps_{rkq_t} \sum_{m=-n}^n \veps{(n)i_1\ldots
  i_n}(m)^*  
 \veps{(n)q_1\ldots q_n}(m) 
O_{q_1\ldots q_{t-1}rq_{t+1}\ldots q_n}
\\
&=\left(\sum_{m=-n}^n \veps{(n)i_1\ldots i_n}(m)^* 
 \veps{(n)q_1\ldots q_n}(m) \right)
\left(
\sum_{t=1}^n i\eps_{rkq_t} 
O_{q_1\ldots q_{t-1}rq_{t+1}\ldots q_n}
\right)
=
\sum_{t=1}^n i\eps_{ki_tr} 
O_{i_1\ldots i_{t-1}ri_{t+1}\ldots i_n},
\end{align*}
where the second equality holds because $O_{nm}$ is a \sito, the
third one by (\ref{eq:cartj}) and the fourth one by
(\ref{eq:spin.matrix.element}).  Notice that, in the next-to-last
equality, the sum in the last parentheses is totally symmetric in
$q_1\ldots q_n$, because $O_{q_1\ldots q_n}$ is.  It is traceless in,
say, $q_t$, $q_1$ because $\eps_{rkq_t}$ is antisymmetric in $q_t$,
$r$. Thus, by total symmetry it is totally traceless, therefore
irreducible (with respect to $q_1\ldots q_n$).  It is therefore left
unchanged by contraction with the first factor, which is a projector
onto the subspace of irreducible tensors.  The last equality then
follows.  We have, thus, proved that $O_{i_1\ldots i_n}$ is a
cartesian tensor operator relative to $\vec{J}$. Since its
irreducibility is obvious by construction, it is a \cito.
\end{proof}

Furthermore, the \sitos\ obtained from Lemma \ref{lem:sito.from.cito}
and the \citos\ from Lemma \ref{lem:cito.from.sito} are actually all
possible ones.  As we now show, there are no more \sitos\ and \citos\
than those described in the Lemmas.
\begin{crl}
  \label{crl:sito.from.cito}
  $O_{nm}$ is a rank-$n$ \sito\ relative to $\vec{J}$ if and only if
  there exists a rank-$n$ cartesian tensor operator $O_{i_1\ldots
    i_n}$ relative to $\vec{J}$, not necessarily irreducible, such
  that $O_{nm}=\veps{(n)i_1\ldots i_n}(m) O_{i_1\ldots i_n}$.  If
  $O_{i_1\ldots i_n}$ is required to be irreducible, it is unique.
\end{crl}
\begin{proof}
If $O_{i_1\ldots i_n}$ is a cartesian tensor operator and
$O_{nm}=\veps{(n)i_1\ldots i_n}(m) O_{i_1\ldots i_n}$, then $O_{nm}$
is a \sito\ by Lemma \ref{lem:sito.from.cito}.  On the other hand,
if $O_{nm}$ is a \sito\ then $O_{i_1\ldots i_n}=\sum_{m'}
\veps{(n)i_1\ldots i_n}(m')^* O_{nm'}$ is a \cito\ by Lemma
\ref{lem:cito.from.sito}, and $O_{nm}=\veps{(n)i_1\ldots i_n}(m)
O_{i_1\ldots i_n}$ by the orthonormality of $\veps{(n)i_1\ldots
  i_n}(\mu)$ with $-n\leq\mu\leq n$. The existence statement is
therefore proved.

Now let $O_{nm}$ and $O_{i_1\ldots i_n}$ be as in the Corollary,
$A_{i_1\ldots i_n}\neq 0$ such that $A_{\{i_1\ldots i_n\}_0}=0$, and 
$O'_{i_1\ldots i_n}=O_{i_1\ldots i_n}+A_{i_1\ldots i_n}$.  Clearly,
$O'_{i_1\ldots i_n}\neq O_{i_1\ldots i_n}$ and $\veps{(n)i_1\ldots
  i_n}(m) O_{i_1\ldots i_n} = \veps{(n)i_1\ldots i_n}(m) O'_{i_1\ldots
  i_n}$, so $O_{i_1\ldots i_n}$ is in general not unique.  If
$O_{i_1\ldots i_n}$ and $O'_{i_1\ldots i_n}$ both satisfy the
hipothesis and are irreducible, however, then their difference is also
irreducible and $\veps{(n)i_1\ldots i_n}(m) (O_{i_1\ldots
  i_n}-O'_{i_1\ldots i_n}) = O_{nm}-O_{nm}=0$
for all $-n\leq m\leq n$, so by completeness $O_{i_1\ldots
  i_n}-O'_{i_1\ldots i_n}=0$ and therefore $O_{i_1\ldots i_n}$ is
unique. 
\end{proof}
\begin{crl}
  \label{crl:cito.from.sito}
$O_{i_1\ldots i_n}$ is a rank-$n$ \cito\ relative to $\vec{J}$ if
and only if there exists a (unique) rank-$n$ \sito\ $O_{nm}$ relative
to $\vec{J}$ such that $O_{i_1\ldots i_n}=\sum_{m=-n}^n
\veps{(n)i_1\ldots i_n}(m)^* O_{nm}$. 
\end{crl}
\begin{proof}
If $O_{nm}$ is a \sito\ then $O_{i_1\ldots i_n}= \sum_{m}
\veps{(n)i_1\ldots i_n}(m)^* O_{nm}$ is a rank-$n$ \cito\ relative to
$\vec{J}$, by Lemma \ref{lem:cito.from.sito}.  Conversely, if
$O_{i_1\ldots i_n}$ is a \cito\ then $O_{nm}=\veps{(n)i_1\ldots
  i_n}(m)O_{i_1\ldots i_n}$ is a \sito, by Lemma
\ref{lem:sito.from.cito}, and $\sum_{m} \veps{(n)j_1\ldots j_n}(m)^*$
$O_{nm} = \sum_{m} \veps{(n)j_1\ldots j_n}(m)^* \veps{(n)i_1\ldots i_n}(m)
O_{i_1\ldots i_n} = O_{j_1\ldots j_n}$, where the last equality
follows by completeness.  Thus, existence of $O_{nm}$ is necessary
and sufficient, as stated.

If $O_{nm}$ and $O'_{nm}$ are \sitos\ both satisfying the statement,
then $O_{nm}=\veps{(n)i_1\ldots i_n}(m) O_{i_1\ldots i_n}=O'_{nm}$ for
all $-n\leq m\leq n$.
\end{proof}

We call a \cito\ $O_{i_1\ldots i_n}$ and a \sito\ $O_{nm}$ related
to each other as described in the Lemmas ``dual'' to each other.  It
is easily shown, using Corollaries \ref{crl:sito.from.cito} and
\ref{crl:cito.from.sito}, that $O_{i_1\ldots i_n}=O_{i_1\ldots
  i_n}^\dagger$ if and only if $O_{nm}=(-1)^m O^\dagger_{n(-m)}$.  

As a simple illustration of the results presented in this section,
consider a single spinless particle moving in a central potential. For
that system any rank-$n$ tensor operator must be a linear combination
of the operators
\begin{equation}
  \label{eq:tnnn}
  T^{(n_r,n_p,n_L)}_{i_1\ldots i_n} = r_{i_1}\ldots r_{i_{n_r}}
p_{i_{n_r+1}}\ldots p_{i_{n_r+n_p}}L_{i_{n_r+n_p+1}}\ldots L_{i_{n}},
\quad n_r,n_p,n_L\geq 0,
\quad n_r+n_p+n_L =n.
\end{equation}
For each triple $(n_r,n_p,n_L)$ we have a different cartesian tensor
operator relative to $\vec{L}$.  None of them is irreducible if $n\geq
2$.  Thus, every rank-$n$ \cito\ is a linear combination of
$T^{(n_r,n_p,n_L)}_{\{i_1\ldots i_n\}_0}$, and every rank-$n$ \sito\
$O_{nm}$ can be expressed as a linear combination of
$\veps{(n)i_1\ldots i_n}(m) T^{(n_r,n_p,n_L)}_{i_1\ldots i_n}$.  In
section \ref{sec:standard} below we briefly describe the tensor
structure of the most common \sitos.

\section{The Wigner--Eckart theorem for irreducible cartesian tensor
  operators} 
\label{sec:we}

If $O_{nk}$ is a rank-$n$ \sito\ relative to $\vec{J}$ the
Wigner--Eckart theorem states that its matrix elements are given by
\begin{equation}
\label{eq:WE.sito}
  \langle j',m'|O_{nk}|j,m\rangle = \langle j'||O_{n}||j\rangle
  \CG{j}{m}{n}{k}{j'}{m'}, 
\end{equation}
where the reduced matrix element $\langle j'||O_{n}||j\rangle$ is
independent of $m$, $m'$. From the Wigner--Eckart theorem for \sitos\
(\ref{eq:WE.sito}) and Corollary \ref{crl:cito.from.sito} we
immediately obtain
\begin{thm*}[Wigner--Eckart theorem for \citos]\label{thm:we.cito}
Let $O_{i_1\ldots i_n}$ be a rank-$n$ \cito\ relative to $\vec{J}$.
Then, its matrix elements are given by
\begin{equation*}
  \langle j',m'|O_{i_1\ldots i_n}|j,m\rangle = 
  \langle j'||\veps{(n)}\cdot O||j\rangle
  \CG{j}{m}{n}{m'-m}{j'}{m'} 
  \veps{(n)i_1\ldots i_n}(m'-m)^*, 
\end{equation*}
where the reduced matrix element $\langle j'||\veps{(n)k_1\dots
  k_n}O_{k_1\dots k_n}||j\rangle$ depends on the operator $O$ and on
$n$, $j'$, $j$, but is independent of $m'$, $m$, and of $i_1\ldots
i_n$.
\end{thm*}
If the \sito\ $O_{nm}$ and the \cito\ $O_{i_1\ldots i_n}$ are dual to
each other in the sense of section \ref{sec:csitos}, then it is clear
that $\langle j'||O_{n}||j\rangle = \langle j'||\veps{(n)}\cdot
O||j\rangle$.  The Wigner--Eckart theorem for \citos\ applies to any
scalar and cartesian vector operators, since for those the
irreducibility requirement is moot.  It can be extended to rank-$n\geq
2$ reducible tensor operators by expanding them in their irreducible
components of rank $\widetilde{n}$ with $0\leq \widetilde{n}\leq n$
and applying the theorem to each irreducible component separately.  To
the best of our knowledge, however, no such general decomposition
theorem has been given in the literature.  The simplest example is, of
course, that of rank-2 cartesian tensors, whose decomposition is
trivial to obtain,
\begin{equation}
\label{crl:we.rank2}  
\begin{gathered}
    \begin{split}
\langle j',m'|O_{h_1h_2}&|j,m\rangle =   
\langle j'||O_{(2)}||j\rangle \CG{j}{m}{2}{m'-m}{j'}{m'}\,
\veps{(2)h_1h_2}(m'-m)^* \\
&+\langle j'||O_{(1)}||j\rangle \CG{j}{m}{1}{m'-m}{j'}{m'}\,
\eps_{h_1h_2k}\veps{(1)k}(m'-m)^* 
+\langle j||O_{(0)}||j\rangle \delta_{j'j}\delta_{m'm}
\delta_{h_1h_2}, 
    \end{split}\\
  O_{(2)i_1i_2} = \frac{1}{2} O_{\{i_1i_2\}_0},
\quad
  O_{(1)h} = \frac{1}{2} \eps_{hk_1k_2} O_{k_1k_2},
\quad
  O_{(0)} = \frac{1}{3} O_{kk}.
\end{gathered}
\end{equation}
Here, the operators $O_{(q)}$, $0\leq q\leq2$, are rank-$q$ \citos\
relative to $\vec{J}$.  We wrote their reduced matrix elements without
$\veps{}$ tensors for brevity, since no notational ambiguity may arise
in this case.  A generalization of the above theorem to the case of
totally--symmetric reducible tensor operators of any rank is given
below in section \ref{sec:totsym}, and to partially irreducible tensor
operators in section \ref{sec:parirr}.

\section{Standard spherical irreducible tensor operators}
\label{sec:standard}

In this section we consider some standard \sitos\ from the
point of view of the preceding sections, and discuss their relation to
\citos.  As a by-product, some relations among standard \sitos\ are found that
may be more difficult to obtain by other methods, as shown in this
section and in section \ref{sec:dershy}.

\subsection{Spherical harmonics}
\label{sec:shy}

Spherical harmonics, as is easy to prove \cite{bou1}, can be written
as 
\begin{equation}
  \label{eq:shyfinal}
  Y_{\ell m}(\verr{}) = N_\ell \veps{(\ell)i_1\ldots i_\ell}(m) \verr{i_1}
  \ldots \verr{i_\ell},
\qquad
N_\ell= \frac{1}{\sqrt{4\pi}} \sqrt{\frac{(2\ell+1)!!}{\ell!}}.
\end{equation}
This equation shows that the \cito\ dual to $Y_{\ell m}$ in the sense
of Lemma \ref{lem:sito.from.cito} is $\verr{(i_1}\ldots
\verr{i_\ell)_0}$, up to a multiplicative constant. The reduced matrix
element of $Y_{\ell m}$ is well-known from the literature
\cite{gal90,var88}, so from (\ref{eq:shyfinal}) we find
\begin{equation}
  \label{eq:shy.reduced}
\langle \ell'
||\veps{(n)i_1\ldots i_n} \verr{i_1} \ldots \verr{i_n}||
\ell\rangle = \frac{1}{N_n} \langle \ell' || Y_{n} ||
\ell\rangle, 
\qquad
  \langle \ell' || Y_{n} || \ell\rangle =
  \sqrt{\frac{(2\ell+1)(2n+1)}{4\pi(2\ell'+1)}} 
  \CG{\ell}{0}{n}{0}{\ell'}{0},
\end{equation}
From (\ref{eq:shyfinal}) and the second equality in
(\ref{eq:std.tens.s.prop}) we recover the familiar relation $Y_{\ell
  m}(\verr{})^*=(-1)^m Y_{\ell(-m)}(\verr{})$.  We notice that
(\ref{eq:std.tens.s.impl}) follows immediately from
(\ref{eq:shyfinal}).  The relation inverse to (\ref{eq:shyfinal}) is
given by (\ref{eq:shy.inverse}) below.  Substituting
(\ref{eq:shyfinal}) in the addition theorem for spherical harmonics
leads to
\begin{equation}
  \label{eq:legendre}
  P_\ell(\verr{}\cdot\verr{}\,') =  \frac{(2\ell-1)!!}{\ell!}
\ver{r}_{(i_1}\ldots\ver{r}_{\,i_\ell)_0} 
\ver{r}\,'_{(i_1}\ldots\ver{r}\,'_{\,i_\ell)_0}, 
\end{equation}
which gives a multilinear representation for Legendre polynomials.

The expressions (\ref{eq:shyfinal}) for $Y_{\ell m}$ and
(\ref{eq:legendre}) for $P_\ell$ as multilinear forms on the unit
sphere have useful applications, some of which we discuss below.
Here, we briefly mention that (\ref{eq:shyfinal}) yields the numerical
coefficients in Stevens' operator-equivalent method \cite{ste52}, as
shown by the relation
\begin{equation}
  \label{eq:stvns1}
  \langle j,m'|\veps{(n)i_1\ldots i_n}(k) r_{i_1}\ldots r_{i_n} |
  j,m\rangle = \frac{\langle j||\veps{(n)j_1\ldots j_n} r_{j_1}\ldots
    r_{j_n}||j\rangle}{\langle j||\veps{(n)k_1\ldots k_n} J_{k_1}\ldots
    J_{k_n}||j\rangle} 
  \langle j,m'|\veps{(n)h_1\ldots h_n}(k) J_{h_1}\ldots J_{h_n} |
  j,m\rangle,
\end{equation}
with the reduced matrix elements given by (\ref{eq:shy.reduced}) and
\cite{bou1} 
\begin{equation}
  \label{eq:stvns2}
\langle j||\veps{(n)k_1\ldots k_n} J_{k_1}\ldots J_{k_n}||j\rangle = 
2^{-n}\sqrt{\frac{n!}{(2n-1)!!}}
\sqrt{\frac{(2j+n+1)!}{(2j+1)(2j-n)!}}. 
\end{equation}
In the simplest case $n=2$ ($n=1$ being trivial) from
(\ref{eq:stvns1}) we get
\begin{equation}
  \label{eq:stvns3}
  \begin{gathered}
2z^2-(x^2+y^2) = 3z^2-r^2 = r^2 C(j,2) (2J_z^2-J_x^2-J_y^2)= C(j,2)
(3J_z^2-j(j+1)), \\
r_i r_j = r^2 C(j,2) \frac{1}{2} (J_i j_j + J_j J_i), \quad i\neq j,\\
C(j,2) = -4 \sqrt{j(j+1)} \sqrt{\frac{2j+1}{(2j-1)(2j+3)}} 
\sqrt{\frac{(2j-2)!}{(2j+3)!}} .
  \end{gathered}
\end{equation}
Further discussion of the operator-equivalent method is outside the
scope of this paper, however, so we refer to \cite{ste52,hof91}. 

\subsubsection{An explicit expression for $Y_{\ell m}$}
\label{sec:xplct}

Taking (\ref{eq:shyfinal}) as a definition of spherical harmonics
leads to an explicit expression for them.  From (\ref{eq:shyfinal}) we
easily obtain $Y_{1m}(\verr{})$ in terms of the spherical coordinates
$\theta$, $\varphi$ of $\verr{}$.  On the other hand, from
(\ref{eq:shyfinal}) and (\ref{eq:std.tens.s.alt}) we obtain
\begin{equation}
  \label{eq:shyfinal.alt}
Y_{\ell m}(\verr{}) = \sqrt{\frac{2\ell+1}{4\pi}} \left(
\sqrt{\frac{4\pi}{3}}\right)^\ell \frac{1}{\ell!} 
\sqrt{(\ell+m)!(\ell-m)!} 
\sum_{\substack{s_1,\ldots,s_n=-1\\s_1+\ldots+s_n=m}}^1
\frac{1}{(\sqrt{2})^{\sum_{h=1}^n|s_h|}} 
Y_{1s_1}(\verr{})\ldots Y_{1s_\ell}(\verr{}).
\end{equation}
This expression for $Y_{\ell m}$ in terms of $Y_{1s}$ can be put in a
more explicit form by exploiting the total symmetry of the summand
under permutations of the summation indices.  We temporarily assume
$m>0$ for concreteness, and define $N_{\pm1,0}$ as the number of 1s,
-1s and 0s, respectively, in the summation multiindex
$(s_1,\ldots,s_\ell)$ in (\ref{eq:shyfinal.alt}).  Thus,
$N_1+N_{-1}+N_0=\ell$, $N_1-N_{-1}=m$, and therefore
$N_0=\ell+m-2N_1$.  It is also easy to see, by considering a few
particular cases, that $m\leq N_1 \leq [(\ell+m)/2]$, where $[\ldots]$
denotes integer part.  There are $\binom{\ell}{N_1}$ ways of
distributing $N_1$ 1s among $\ell$ indices $s_i$, and there are
$\binom{\ell-N_1}{N_1-m}$ ways to distribute $N_{-1}=N_1-m$ -1s among
the remaining $\ell-N_1$ indices, and the remaining $N_0$ indices must
take the value 0.  Thus, for given $\ell$, $m$ we can reduce the sum
in (\ref{eq:shyfinal.alt}) to a single sum over $N_1$.  By using the
explicit form of $Y_{1s}$, from (\ref{eq:shyfinal.alt}) we obtain
\begin{equation}
  \label{eq:shyfinal.alt.2}
  \begin{aligned}
Y_{\ell m}(\verr{}) &= \sqrt{\frac{2\ell+1}{4\pi}}
\sqrt{\frac{(\ell-m)!}{(\ell+m)!}} e^{im\varphi} P_{\ell
  m}(\cos\theta), \\
P_{\ell m}(x) &= \frac{(\ell+m)!}{\ell!} \sum_{N_1=m}^{[(\ell+m)/2]} 
\binom{\ell}{N_1} \binom{\ell-N_1}{N_1-m}
\frac{(-1)^{N_1}}{2^{2N_1-m}} (\sqrt{1-x^2})^{2N_1-m}
x^{\ell-2N_1+m},
  \end{aligned}
\end{equation}
valid for integer $\ell\geq0$ and $-\ell\leq m\leq\ell$.  We see that
starting from (\ref{eq:shyfinal}) we not only recovered the standard
expression (\ref{eq:shy.std}) but also obtained an explicit expression
for the associated Legendre function $P_{\ell m}$ by purely tensorial
considerations without reference to the Legendre differential
equation. (For further expressions for $P_{\ell m}$ and $Y_{\ell m}$
see \cite{var88,bat}.  See \cite{tor13} for related recent results.)

\subsection{Bipolar spherical harmonics}
\label{sec:bshy}

Consider a system formed by two spinless particles moving in a central
potential, with orbital angular momenta $\vec{L}_{1,2}$ coupled to
total angular momentum $\vec{J}$. Its
angular wave function is given by a bipolar spherical harmonic
$Y^{\ell_1\ell_2}_{j m} (\verr{1},\verr{2})$, defined as \cite{var88}
\begin{equation}
  \label{eq:bipolar}
  Y^{\ell_1\ell_2}_{j m} (\verr{1},\verr{2})
\equiv \langle \verr{1}, \verr{2}| 
\ell_1, \ell_2, j, m \rangle =
\sum_{\mu_1=-\ell_1}^{\ell_1}\sum_{\mu_2=-\ell_2}^{\ell_2} 
\CG{\ell_1}{\mu_1}{\ell_2}{\mu_2}{j}{m} 
Y_{\ell_1\mu_1}(\verr{1})Y_{\ell_2\mu_2}(\verr{2}).
\end{equation}
Its complex conjugation properties follow from (\ref{eq:bipolar}),
$Y^{\ell_1\ell_2}_{\ell m}(\verr{1},\verr{2})^* =
(-1)^{\ell_1+\ell_2-\ell}(-1)^m
Y^{\ell_1\ell_2}_{\ell(-m)}(\verr{1},\verr{2}) =(-1)^m
Y^{\ell_2\ell_1}_{\ell(-m)}(\verr{2},\verr{1})$.  
From the Clebsch--Gordan coupling of two spherical harmonics
\cite{var88} we have the equality
\begin{equation}
  \label{eq:bipolar.local}
  Y^{\ell_1\ell_2}_{j m} (\verr{},\verr{}) = 
\sqrt{\frac{(2\ell_1+1)(2\ell_2+1)}{4\pi(2j+1)}} 
\CG{\ell_1}{0}{\ell_2}{0}{j}{0} Y_{jm}(\verr{}),
\end{equation}
that we will need below.  Bipolar spherical harmonics define operators
in the Hilbert state--space of the two--particle system, acting
multiplicatively in the coordinate representation
$\langle\verr{1},\verr{2} | Y^{\ell_1\ell_2}_{jm} | \psi\rangle =
Y^{\ell_1\ell_2}_{jm}(\verr{1},\verr{2}) \psi(\verr{1},\verr{2})$. An
analogous, but different, operator is obtained in the momentum
representation by the multiplicative action of
$Y^{\ell_1\ell_2}_{jm}(\verp{1},\verp{2})$. As operators, bipolar
spherical harmonics are \sitos\ of rank $j$ relative to $\vec{J} =
\vec{L}_1 + \vec{L}_2$.

If the coupling in (\ref{eq:bipolar}) is maximal, $j=\ell+\ell'$, then
$Y^{\ell\ell'}_{j m}$ has a tensorial representation analogous to
(\ref{eq:shyfinal}) that follows from (\ref{eq:maximal.coupling})
\begin{equation}
  \label{eq:max.bshy}
Y^{\ell\ell'}_{(\ell'+\ell) m} (\verr{},\ver{r}\,') = 
N_{\ell} N_{\ell'} \veps{(j)i_1\ldots i_{\ell} j_1\ldots j_{\ell'}}(m) 
\verr{i_1}^{}\ldots \verr{i_\ell}^{} \ver{r}\,'_{j_1}\ldots
\ver{r}\,'_{j_{\ell'}}. 
\end{equation}
Thus, a maximally-coupled bipolar spherical harmonic is dual to the
cartesian irreducible tensor operator $N_{\ell} N_{\ell'}
\verr{\{i_1}^{}\ldots \verr{i_\ell}^{}\ver{r}\,'_{i_{\ell+1}}\ldots
\ver{r}\,'_{i_{\ell+\ell'}\}_0}$. We have, in particular, 
\begin{equation}
  \label{eq:stdtens.compact}
  \veps{(2)ij}(m) = \frac{4\pi}{3} Y^{11}_{2m}(\ver{e}^i,\ver{e}^j)
\end{equation}
which is the rank-2 analog of (\ref{eq:stdvec.b}).

If the angular-momentum coupling in $Y^{\ell\ell'}_{j m}
(\verr{},\ver{r}\,')$ is not maximal its tensorial expression is more
complicated.  It has been derived in the general case
$|\ell'-\ell|\leq j\leq\ell'+\ell$ in \cite{bou2} using the methods of
\cite{bou1}, which are beyond the scope of this paper.  We quote the
result without proof in (\ref{eq:bipolar.gen.not}) for completeness,
and because we need some of its consequences below. From that
expression we can read off the cartesian irreducible tensor dual to
$Y^{\ell\ell'}_{j m} (\verr{},\ver{r}\,')$, which is the irreducible
part of the tensor contracted with $\veps{(j)i_1\ldots i_j}$ there. 

\subsubsection{The binomial expansion for spherical harmonics}
\label{sec:binomial}

Given two position vectors $\vec{r}_a$, $a=1,2$, it may be of interest
to compute $Y_{\ell m}(\verr{})$, with $\vec{r}$ a linear combination
of $\vec{r}_{1,2}$, in terms of $Y_{\ell m}({\verr{1,2}})$.  For
instance, $\vec{r}$ may be the center-of-mass of $\vec{r}_{1,2}$ or
their associated relative position vector $\vec{r}_1-\vec{r}_2$, or,
in the momentum representation, the total or relative momentum of two
particles.  Setting
$\verr{}=(\alpha\vec{r}_1+\beta\vec{r}_2)/|\alpha\vec{r}_1+\beta\vec{r}_2|$
in (\ref{eq:shyfinal}), with $\alpha$, $\beta$ real numbers, applying
the binomial expansion for the multiple product of $\verr{}$ there,
and using (\ref{eq:maximal.coupling}), leads to the binomial expansion
for spherical harmonics
\begin{equation}
  \label{eq:binomial.shy.final}
  Y_{\ell m}\left(\frac{\alpha\vec{r}_1+\beta\vec{r}_2}
{|\alpha\vec{r}_1+\beta\vec{r}_2|}\right) = 
\sqrt{4\pi}
\sum_{n=0}^\ell \frac{1}{\sqrt{2(\ell-n)+1}}
\sqrt{\binom{2\ell+1}{2n+1}}
\frac{(\alpha|\vec{r}_1|)^n(\beta|\vec{r}_2|)^{\ell-n}}
{|\alpha\vec{r}_1+\beta\vec{r}_2|^\ell}
Y^{n(\ell-n)}_{\ell m}(\verr{1},\verr{2}).
\end{equation}
Notice that the bipolar spherical harmonic in
(\ref{eq:binomial.shy.final}) is maximally coupled.  The binomial
expansion can be extended to a multinomial expansion for $Y_{\ell
  m}(\vec{R}/|\vec{R}|)$ with $\vec{R}=\sum_{a=1}^k \alpha_a
\vec{r}_a$, $k\geq2$. In that case, because all angular-momentum
couplings are maximal as in (\ref{eq:binomial.shy.final}), all
coupling schemes lead to the same result.

\subsection{Tensor spherical harmonics}
\label{sec:tshy}

Tensor spherical harmonics are simultaneous eigenfunctions of
$\vec{L}^2$, $\vec{S}^2$, $\vec{J}^2$, and $J_3$, as is appropriate to
the wave functions of a particle with spin $\vec{S}$ and orbital
angular-momentum $\vec{L}$ coupled to total angular-momentum
$\vec{J}$.  They are defined as \cite{bou2}
\begin{equation}
  \label{eq:tshy}
  \left(Y^{\ell n}_{jm}(\verr{})\right)_{i_1\ldots i_n} =
  \sum_{\mu=-\ell}^\ell \sum_{\nu=-n}^n \CG{\ell}{\mu}{n}{\nu}{j}{m} 
  Y_{\ell\mu}(\verr{}) \veps{(n)i_1\ldots i_n}(\nu).
\end{equation}
From the second equality in (\ref{eq:std.tens.s.prop}) and the
conjugation of spherical harmonics we derive the conjugation relation
$(Y^{\ell s}_{j m}(\verr{}))_{i_1\ldots i_s}^* = (-1)^{\ell+s-j}
(-1)^m (Y^{\ell s}_{j(-m)}(\verr{}))_{i_1\ldots i_s}$.  The quantities
$Y^{\ell n}_{jm}(\verr{})$ are tensor functions defined on the unit
sphere, known as spin-$n$ spherical harmonics or rank-$n$ tensor
spherical harmonics.  If $n=0$ we have that $Y^{\ell 0}_{\ell
  m}=Y_{\ell m}$ is an ordinary, scalar spherical harmonic.  For
$n=1$, $Y^{\ell 1}_{j m}(\verr{})$ is a vector spherical harmonic.  In
that case (\ref{eq:tshy}) agrees with the definition given in
\cite{gal90,ynd96}, agrees with \cite{new82} up to a factor of $i$,
and differs from those in \cite{mes61,var88} in the choice of spin
wave-function basis.  $Y^{\ell s}_{j m}(\verr{})$ transforms as a
rank-$s$ cartesian tensor relative to $\vec{J}$.  Its total
contraction with an orbital rank-$s$ tensor operator transforms as a
rank-$j$ \sito\ relative to $\vec{L}$, scalar relative to $\vec{S}$. 

Maximally coupled tensor spherical harmonics ($j=\ell+n$) have a
simple tensorial expression analogous to (\ref{eq:shyfinal}).  By
substituting (\ref{eq:shyfinal}) in (\ref{eq:tshy}) and using
(\ref{eq:maximal.coupling}) we get 
\begin{equation}
  \label{eq:maxtshy}
  \left(Y^{\ell n}_{(\ell+n)m}(\verr{})\right)_{i_1\ldots i_n} =  
N_\ell \veps{(\ell+n)i_1\ldots i_{\ell+n}} \verr{i_{n+1}} \ldots
\verr{i_{n+\ell}}. 
\end{equation}
In particular we have the relations
\[
(Y^{0s}_{sm}(\verr{}))_{i_1\ldots 
i_s}=\frac{1}{\sqrt{4\pi}} \veps{(s)i_1\ldots i_s}(m),
\qquad
\left(Y^{\ell s}_{(\ell+s)m}(\verr{})\right)_{i_1\ldots i_s}
\verr{i_s}=   \sqrt{\frac{\ell+1}{2\ell+3}} 
\left(Y^{(\ell+1)(s-1)}_{(\ell+s)m}(\verr{})\right)_{i_1\ldots
  i_{s-1}}, 
\]
and $\left(Y^{\ell n}_{(\ell+n)m}(\verr{})\right)_{i_1\ldots i_n}
\verr{i_{1}} \ldots \verr{i_{n}} = (N_\ell/N_{\ell+n}) Y_{(\ell+n)
  m}(\verr{})$ is an ordinary spherical harmonic.  The tensorial
expression of non-maximally coupled tensor spherical harmonics is more
involved than that of maximally coupled ones.  We can get some insight
into it by relating tensor spherical harmonics to bipolar ones and using
(\ref{eq:bipolar.gen.not}).  From (\ref{eq:shyfinal}),
(\ref{eq:bipolar}) and (\ref{eq:tshy}) we obtain the bipolar spherical
harmonics as
\begin{equation}
  \label{eq:tensor.from.bipolar}
 Y^{\ell s}_{jm}(\verr{},\verrp{})=
N_s  \left(Y^{\ell s}_{jm}(\verr{})\right)_{i_1\ldots i_s}
  \ver{r}\,'_{i_1}\ldots\ver{r}\,'_{i_s}.
\end{equation}
This relation can be inverted to write tensor spherical harmonics in
terms of bipolar ones
\begin{equation}
  \label{eq:tensor.from.bipolar.b}
  \left(Y^{\ell s}_{jm}(\verr{})\right)_{i_1\ldots i_s} =
\sqrt{4\pi} \sqrt{\frac{1}{s!(2s+1)!}}   (-1)^{\ell+s-j}
\partial'_{i_1}\ldots\partial'_{i_s} 
\left(|\vec{r}\,'|^s Y^{s\ell}_{jm}(\ver{r}\,',\verr{})\right).
\end{equation}
From this equation and (\ref{eq:bipolar.gen.not}) the tensorial
expressions of non-maximally coupled tensor spherical harmonics can be
obtained.  More interestingly, by writing the cartesian tensors
appearing in (\ref{eq:bipolar.gen.not}) as maximally coupled tensor
spherical harmonics, through (\ref{eq:maxtshy}), we can express
non-maximally-coupled tensor spherical harmonics in terms maximally
coupled ones. Once those relations have been obtained, use of
(\ref{eq:maxtshy}) yields the sought-for tensorial expressions. We
will restrict ourselves here to quoting the results for next-- and
next--to--next--to--maximal coupling.  For $j=\ell+s-1$ we have
\begin{equation}
  \label{eq:NMC}
  Y^{\ell s}_{(\ell+s-1)m}(\verr{}) =
  -\sqrt{\frac{2\ell+1}{s(\ell+s)}}
  \left(\verr{}\cdot\vec{S}_{(s)}\right)\cdot Y^{(\ell-1)
    s}_{(\ell+s-1)m}(\verr{}), 
\end{equation}
Notice that in the momentum representation the matrix on the
right-hand side would be the helicity operator $\verp{}\cdot\vec{S}$.  
Similarly, for $j=\ell+s-1$ and $s\geq2$  we have
\begin{subequations}
  \label{eq:NNMC}
\begin{equation}
  \label{eq:NNMC.2}
  \begin{aligned}
  Y^{\ell s}_{(\ell+s-2)m}(\verr{}) &=  \frac{2\ell-1}{\sqrt{s(2s-1)}}
  \sqrt{\frac{2\ell+1}{\ell-1}}\frac{1}{\sqrt{2(\ell+s)-1}\sqrt{\ell+s-1}} 
\\
&\quad\times
\left(\rule{0pt}{14pt}
\left(\verr{}\cdot\vec{S}_{(s)}\right)\cdot
\left(\verr{}\cdot\vec{S}_{(s)}\right)  
-\frac{s(\ell+s-1)}{2\ell-1}\right)
\cdot  Y^{(\ell-2)s}_{(\ell+s-2)m}(\verr{}),
  \end{aligned}
\end{equation}
and, if $s=1$,
\begin{equation}
  \label{eq:NNMC1}
  Y^{\ell 1}_{(\ell-1)m}(\verr{}) = \sqrt{\frac{\ell-1}{\ell}} 
  Y^{(\ell-2) 1}_{(\ell-1)m}(\verr{}) - \sqrt{\frac{2\ell-1}{\ell}} 
  Y_{(\ell-1)m}(\verr{})\, \verr{}.
\end{equation}
\end{subequations}
Analogous relations can be obtained for $j=\ell+s-\nu$ with $\nu\geq
3$.  Additionally, such relations as (\ref{eq:NMC}) and
(\ref{eq:NNMC}) are of interest because they are independent of the
spin wave-function basis used to define the tensor spherical
harmonics.  For brevity, however, we will not dwell longer on those
issues here. 

\subsection{Spin polarization operators}
\label{sec:spin.polar}

In the $(2s+1)$--dimensional space of spin states of a spin-$s$
particle, the spin operator $\vec{S}$ may be viewed as a
$(2s+1)\times(2s+1)$ matrix and $\vec{S}^2=s (s+1) I$, with $I$ the
identity matrix.  A complete set of $(2s+1)^2$ matrices in that space
is given by the polarization operators defined in section 2.4 of
\cite{var88} as 
\begin{equation}
  \label{eq:spin.polar}
  T_{\ell m}(s) = \kappa_{\ell}(s)
  \left(\vec{S}\cdot\vec{\nabla}\right)^\ell \left(|\vec{r}\,|^\ell
    Y_{\ell m}(\verr{})\right),
\qquad
\kappa_{\ell}(s)=\frac{2^\ell}{\ell!} \sqrt{\frac{4\pi
    (2s-\ell)!}{(2s+\ell+1)!}},
\qquad
0\leq \ell \leq 2s,
\quad
-\ell \leq m \leq \ell.   
\end{equation}
From this equation and (\ref{eq:std.tens.s.impl}) we immediately
obtain the equivalent expression 
\begin{equation}
  \label{eq:spin.polar.2}
  T_{\ell m}(s) = \kappa'_{\ell}(s)  \veps{(\ell)i_1\ldots i_\ell}(m)
  S_{i_1} \ldots S_{i_\ell},
\quad
\kappa'_{\ell}(s) = 2^\ell
\sqrt{\frac{(2s-\ell)!(2\ell+1)!!}{\ell!(2s+\ell+1)!}}. 
\end{equation}
Thus, $T_{\ell m}(s)$ is a rank-$\ell$ \sito\ relative to $\vec{S}$,
dual to the \cito\ $1/\ell!  S_{\{i_1} \ldots S_{i_\ell\}_0}$ (i.e.,
the irreducible component of the cartesian tensor operator $S_{i_1}
\ldots S_{i_\ell}$) introduced in \cite{zem65}.  The reduced matrix
element of $T_{\ell m}(s)$ is $\langle s ||T_{\ell m}(s)||s \rangle =
\sqrt{(2\ell+1)/(2s+1)}$.  The operator $T_{\ell m}(s)$ is hermitian,
$T_{\ell m}(s)^\dagger = (-1)^m T_{\ell (-m)}(s)$, and satisfies the
orthonormality relation
\begin{equation}
  \label{eq:spin.polar.3}
  \mathrm{Tr} \left(T_{\ell' m'}(s)^\dagger T_{\ell m}(s) \right)
  \equiv 
  \sum_{\mu} \langle s,\mu | T_{\ell' m'}(s)^\dagger T_{\ell m}(s) | s,
  \mu \rangle = \delta_{\ell'\ell} \delta_{m'm},
\end{equation}
which can be readily verified by inserting $\sum_{\mu'}|s,
\mu'\rangle\langle s,\mu'|$ between the two operators and applying the
Wigner--Eckart theorem to the resulting matrix elements.  A  more
detailed description of the properties of the operators $T_{\ell
  m}(s)$ is given in the reference cited above.

\subsection{Electric multipole moments}
\label{sec:emultipole}

As an application of the results of section \ref{sec:shy}, we derive
in this section a general expression for the cartesian $2^n$-pole
electric tensor and its relation to the spherical one.  The magnetic
multipoles are discussed below in section \ref{sec:mmultipole}. The
electric potential is given by \cite{jac99}
\begin{equation}
  \label{eq:potential}
  \phi(\vec{r},t) = \frac{1}{4\pi\eps_0} \int_V d^3r'
  \frac{\rho(\vec{r}\,',t)}{|\vec{r} - \vec{r}\,'|},
\end{equation}
where $\rho(\vec{r},t)$ is the charge density, assumed to vanish for
all $t$ outside a bounded volume $V$.  The Coulomb integral
(\ref{eq:potential}) gives the potential $\phi(\vec{r})$ in
electrostatics, and in electrodynamics in the Coulomb gauge with
$\rho(\vec{r},t)$ the instantaneous charge density.  More generally,
(\ref{eq:potential}) holds when retardation effects can be neglected. 

Assuming $|\vec{r}\,|>|\vec{r}\,'|$ for all $\vec{r}\,'$ in $V$, we
can substitute the well-known expansion of $1/|\vec{r} - \vec{r}\,'|$
in spherical harmonics \cite{jac99} in (\ref{eq:potential}) to obtain
the spherical multipole expansion
\begin{equation}
  \label{eq:sph.multipole}
\phi(\vec{r}) = \frac{1}{4\pi\eps_0} \sum_{n=0}^\infty
\frac{4\pi}{2n+1} \frac{1}{|\vec{r}\,|^{n+1}} \sum_{m=-n}^n q_{nm}^*
Y_{nm}(\verr{}), 
\qquad
q_{nm} = \int_V d^3r' \rho(\vec{r}\,') |\vec{r}\,'|^{n}
Y_{nm}(\verr{}\,').
\end{equation}
In this equation and in what follows we omit the temporal coordinate
in $\phi(\vec{r},t)$, $\rho(\vec{r},t)$ and $q_{nm}(t)$ for brevity. 
For $n\geq 0$ fixed, the $2n+1$ quantities $q_{nm}$ are called
``spherical $2^n$-pole moments'' \cite{jac99}.  In the coordinate
representation in quantum mechanics $q_{nm}$ is a rank-$n$ \sito, as
is apparent from the first equality in (\ref{eq:sph.multipole}) since
$Y_{nm}$ is a rank-$n$ \sito\ and $\phi$ is a scalar.  (Notice that
our $q_{nm}$ is the complex conjugate of the one in \cite{jac99}.)
A cartesian multipole expansion, on the other hand, is of the form 
\begin{equation}
  \label{eq:cart.multipole.exp}
\phi(\vec{r}) = \frac{1}{4\pi\eps_0} \sum_{n=0}^\infty \frac{1}{n!} 
\frac{1}{|\vec{r}\,|^{2n+1}} Q_{i_1\ldots i_n} r_{i_1}\ldots r_{i_n},
\end{equation}
where $Q_{i_1\ldots i_n}$ is the cartesian $2^n$-pole tensor. Clearly,
$Q_{i_1\ldots i_n}$ must be a rank-$n$ cartesian tensor, since $\phi$
in (\ref{eq:cart.multipole.exp}) is a scalar, and it must be
completely symmetric since any antisymmetric part would not contribute
to (\ref{eq:cart.multipole.exp}).  It must also be traceless, because
otherwise the traces of $Q_{i_1\ldots i_n}$ would make contributions
of $\mathcal{O}(1/r^{k+1})$ with $k<n$.  Thus, $Q_{i_1\ldots i_n}$ is
a rank-$n$ cartesian irreducible tensor in the classical theory, and a
rank-$n$ \cito\ in the quantum theory.  To obtain
(\ref{eq:cart.multipole.exp}) and the expression for $Q_{i_1\ldots
  i_n}$ from (\ref{eq:potential}), we start from the expansion of
$1/|\vec{r} - \vec{r}\,'|$ in Legendre polynomials \cite{jac99}
written in the form
\begin{equation}
\label{eq:legendre1}
\frac{1}{|\vec{r}-\vec{r}\,'|} = \sum_{n=0}^\infty
\frac{|\vec{r}\,'|^n}{|\vec{r}\,|^{n+1}} P_n(\verr{}\cdot\verrp{}),  
\end{equation}
which we substitute in (\ref{eq:potential}).  By using the expression 
(\ref{eq:legendre}) for the Legendre polynomial $P_n$ we are led to
(\ref{eq:cart.multipole.exp}) with
\begin{equation}
  \label{eq:cart.multipole}
\frac{1}{(2n-1)!!}  Q_{i_1\ldots i_n} = \int_V d^3r' \rho(\vec{r}\,') r'_{(i_1}
  \ldots r'_{i_n)_0} = 
\sum_{m=-n}^n  \veps{(n)i_1\ldots i_n}(m)^* \veps{(n)j_1\ldots j_n}(m)
\int_V d^3r' \rho(\vec{r}\,') r'_{j_1} \ldots r'_{j_n}.
\end{equation}
As is easy to check, for $n=0$, 1, 2, $Q_{i_1\ldots i_n}$ is the total
charge, dipole vector, and quadrupole tensor, respectively, of the
charge distribution $\rho$.  We can make contact with the notation of
\cite{kie98} by substituting 
\begin{equation}
\verr{(i_1}\ldots\verr{i_n)_0} = (-1)^n \frac{r^{n+1}}{(2n-1)!!} 
\partial_{i_1}\ldots\partial_{i_n}\frac{1}{|\vec{r}\,|} 
\end{equation}
in (\ref{eq:cart.multipole}). From
(\ref{eq:cart.multipole}) and (\ref{eq:shyfinal}) we get
\begin{equation}
  \label{eq:multipole.final}
q_{nm} = \frac{1}{\sqrt{4\pi}} \sqrt{\frac{2n+1}{n!(2n-1)!!}}
Q_{i_1\ldots i_n} \veps{(n)i_1\ldots i_n}(m),
\qquad
Q_{i_1\ldots i_n} = \sqrt{4\pi} 
\sqrt{\frac{n!(2n-1)!!}{2n+1}} \sum_{m=-n}^n q_{nm} \veps{(n)i_1\ldots
  i_n}(m)^*. 
\end{equation}
We see that the rank-$n$ spherical and cartesian multipole tensors are
related to each other according to the general results of section
\ref{sec:csitos}, up to a normalization constant.  Equation
(\ref{eq:multipole.final}) reproduces the expressions for $q_{nm}$ in
terms of $Q_{i_1\ldots i_n}$ given in equations (4.4)---(4.6) of
\cite{jac99} (taking into account that our $q_{nm}$ is the complex
conjugate of the one in that reference) for $n=0$, 1, 2, and
generalizes them to any natural $n$, making completely explicit the
analogous relations obtained in \cite{tor02}.  

\section{The Wigner-Eckart theorem for reducible symmetric cartesian
  tensor operators}
\label{sec:totsym}

The results of the previous sections, and in particular the
Wigner--Eckart theorem, can in principle be extended to generic
reducible tensors and tensor operators by decomposing them into their
irreducible components.  In this section we consider the extension to
the case of totally symmetric reducible cartesian tensors.

We introduce a basis of totally symmetric rank-$n$ tensors as the set of
tensors $\veps{\{n,s\}}(m)$ of rank $n$ and spin $s$, with $s=0,$ 2, \ldots
$n$ if $n$ is even and $s=1,$ 3, \ldots $n$ if it is odd, and $-s\leq
m\leq s$, defined as
\begin{subequations}
  \label{eq:tot.sym.def}  
\begin{equation}
  \label{eq:tot.sym.def.a}
  \veps{\{n,s\}i_1\ldots i_n}(m) = \lambda_{(n,s)} \frac{1}{n!} 
  \veps{(s)\{i_1\ldots i_s}(m) \delta_{i_{s+1}i_{s+2}}\ldots
  \delta_{i_{n-1}i_{n}\}}. 
\end{equation}
It is clear from this definition that for $s=n$ the basis tensor is
irreducible and $\veps{\{n,n\}i_1\ldots i_n}(m) = \veps{(n)i_1\ldots
  i_n}(m)$.  Furthermore, if $n$ is even and $s=0$ then
$\veps{(0)}(0)=1$. The normalization constant in
(\ref{eq:tot.sym.def.a}) is given by
\begin{equation}
  \label{eq:tot.sym.def.b}
  \lambda_{(n,s)} = \left( 
    \frac{1}{2^{n_\delta} n_{\delta}!} \frac{n!}{s!}
    \frac{(2s+1)!!}{(n+s+1)!!} \right)^{1/2},
\qquad
n_\delta = \frac{n-s}{2}.
\end{equation}
\end{subequations}
From (\ref{eq:tot.sym.def}) the basis tensors are seen to satisfy the
normalization and complex conjugation relations
\begin{equation}
  \label{eq:tot.sym.props}
  \veps{\{n,s'\}}(m')^*\cdot \veps{\{n,s\}}(m) = \delta_{s's} \delta_{m'm}, 
\qquad
\veps{\{n,s\}}(m)^* = (-1)^m \veps{\{n,s\}}(-m),
\end{equation}
where the dot stands for total index contraction.  The definition
(\ref{eq:tot.sym.def}) is far from arbitrary; from it and
(\ref{eq:std.tens.s.impl}) we obtain the equality
\begin{equation}
  \label{eq:tot.sym.alt.def}
  \veps{\{n,s\}i_1\ldots i_n}(m) = \lambda'_{(n,s)} \partial_{i_1}
  \ldots \partial_{i_n} \left( |\vec{r}|^n Y_{sm}(\verr{}) \right),
\qquad
\lambda'_{(n,s)} = \left( 4\pi \frac{1}{2^{n_\delta} n_{\delta}!} 
\frac{1}{n!(n+s+1)!!}\right)^{1/2},
\end{equation}
valid for $0\leq n-s$ even, generalizing (\ref{eq:std.tens.s.impl}) to the
case of totally symmetric reducible tensors, and which can be taken as
a definition of $\veps{\{n,s\}}$ equivalent to (\ref{eq:tot.sym.def}).
It is clear from (\ref{eq:tot.sym.def}) that for $n$ fixed there are
$(n+1)(n+2)/2$ basis tensors, spanning the subspace of totally
symmetric tensors.  It is straightforward to show from
(\ref{eq:c.tensors}), (\ref{eq:s.mat.prop.b}) that the basis tensors
$\veps{\{n,s\}}(m)$ are eigenfunctions of $\vec{S}_{(n)}^2$ and
$\ver{z}\cdot\vec{S}_{(n)}$ with eigenvalues $s$ and $m$.

Given any totally symmetric rank-$n$ complex tensor $A_{i_1\ldots
  i_n}$ we have the expansion
\begin{equation}
  \label{eq:tot.sym.expn}
  A_{i_1\ldots i_n}= \sum_{\substack{s=0\\(n-s)\mathrm{even}}}^n
  \lambda_{(n,s)} \sum_{m=-s}^s \left(\rule{0pt}{11pt}
    \mathrm{Tr}_{(s)}(A)\cdot\veps{(s)}(m)\right)
  \veps{\{n,s\}i_1\ldots i_n}(m)^*,
\quad
  \mathrm{Tr}_{(s)}(A)_{i_1\ldots i_s}\equiv A_{i_1\ldots
    i_sk_1k_1\ldots k_{n_\delta}k_{n_\delta}}, 
\end{equation}
with $n_\delta$ as in (\ref{eq:tot.sym.def.b}).  As a particular case
of (\ref{eq:tot.sym.expn}) we have
\begin{equation}
  \label{eq:shy.inverse}
  \verr{i_1}\ldots \verr{i_n} =
  \sum_{\substack{s=0\\(n-s)\mathrm{even}}}^n
  \lambda'_{(n,s)} \sum_{m=-s}^s Y_{sm}(\verr{})
  \veps{\{n,s\}i_1\ldots i_n}(m)^*,
\end{equation}
with $\lambda'_{(n,s)}$ as in (\ref{eq:tot.sym.alt.def}), which is the
relation inverse to (\ref{eq:shyfinal}).  From the expansion
(\ref{eq:tot.sym.expn}) and the Wigner--Eckart theorem for \citos\ we
obtain
\begin{thm*}[Wigner--Eckart theorem for totally symmetric tensor operators]
Let $O_{i_1\ldots i_n}$ be a totally symmetric cartesian tensor
operator relative to $\vec{J}$.  Then, its matrix elements are given
by  
\begin{equation}
  \label{eq:tot.sym.we}
 \langle j',\mu'| O_{i_1\dots i_n} |j,\mu\rangle =
\sum_{\substack{s=0\\(n-s)\,\mathrm{even}}}^n
  \lambda_{(n,s)} 
\langle j' || \veps{(s)}\cdot \mathrm{Tr}_{(s)}(O)||j\rangle
\sum_{m=-s}^s 
\CG{j}{\mu}{s}{m}{j'}{\mu'} \veps{\{n,s\}i_1\ldots i_n}(m)^*. 
\end{equation}
\end{thm*}
Among the most commonly--occurring totally symmetric tensor operators
we have $r_{i_1}\ldots r_{i_n}$ and $p_{i_1}\ldots p_{i_n}$. The
former is the simplest possible example, since its traces are trivial
to compute.  We have
\begin{equation*}
 \langle \ell',\mu'| r_{i_1}\ldots r_{i_n} |\ell,\mu\rangle = 
|\vec{r}|^n \sum_{\substack{s=0\\(n-s)\,\mathrm{even}}}^n
  \lambda_{(n,s)} 
\langle\ell' || \veps{(s)k_1\ldots k_s}
\verr{k_1}\ldots\verr{k_s}||\ell\rangle 
\sum_{m=-s}^s 
\CG{\ell}{\mu}{s}{m}{\ell'}{\mu'}\veps{\{n,s\}i_1\ldots i_n}(m)^*, 
\end{equation*}
with the reduced matrix element given by (\ref{eq:shy.reduced}).  From
this matrix element we obtain a related result involving spherical
harmonics 
\begin{equation*}
  \begin{split}
  \sum_{\mu',\mu} Y_{\ell'\mu'}(\verqp{}) 
 \langle \ell',\mu'| \verr{i_1}\ldots \verr{i_n} |\ell,\mu\rangle 
Y_{\ell\mu}(\verq{})^* =
\sum_{\substack{s=0\\(n-s)\,\mathrm{even}}}^n
(-1)^{\ell'-s} \sqrt{\frac{2\ell'+1}{2s+1}} \lambda_{(n,s)} 
\langle\ell' || \veps{(s)k_1\ldots k_s}
\verr{k_1}\ldots\verr{k_s}||\ell\rangle 
\\\times
\sum_{m=-s}^s 
\veps{\{n,s\}i_1\ldots i_n}(m)^*
Y^{\ell\ell'}_{sm}(\verq{},\verqp{}). 
  \end{split}
\end{equation*}
Notice that the left-hand side of this equality is just the matrix
element $\langle\verqp{}|\P_{\ell'} \verr{i_1}\ldots
\verr{i_n}\P_{\ell}|\verq{}\rangle$, with $\P_\ell$ the
angular-momentum projector operator \cite{bou1}.  Setting
$\verqp{}=\verq{}$ in this last equality and using
(\ref{eq:bipolar.local}) on its right-hand side, we obtain the
expansion in spherical harmonics of $\langle\verq{}|\P_{\ell'}
\verr{i_1}\ldots \verr{i_n}\P_{\ell}|\verq{}\rangle$.  This simple
example illustrates the power of our approach; the case of tensor
powers of the momentum operator, $p_{i_1}\ldots p_{i_n}$, is discussed
in detail in the following section.

\section{Derivatives of spherical harmonics to all orders}
\label{sec:dershy}

The gradient of $Y_{\ell m}(\verr{})$ can be expressed as a linear
combination of vector spherical harmonics, a widely known result
going back to \cite{bet33} and now standard textbook material
\cite{gal90,var88,edm96}.  In this section we generalize that result
to the derivatives of $Y_{\ell m}(\verr{})$ of all orders, expressed
in terms of tensor spherical harmonics, as an application of
(\ref{eq:tot.sym.we}).  For that purpose, we would like to write
equalities of the form 
\begin{equation*}
  \partial_{i_1}\ldots\partial_{i_n} Y_{\ell m}(\verr{}) \neq
  \langle \ver{r} | \partial_{i_1}\ldots\partial_{i_n} | \ell,
  m\rangle \neq
  \sum_{\ell',m'} Y_{\ell' m'}(\verr{}) \langle \ell', m' |
  \partial_{i_1}\ldots\partial_{i_n} | \ell, m\rangle.
\end{equation*}
The reason why we cannot use equal signs in these relations is that,
whereas $Y_{\ell m}(\verr{})$ is independent of $|\vec{r}\,|$,
$\partial_i Y_{\ell m}(\verr{})$ and higher derivatives are not.  We
remark, however, that $|\vec{r}\,|^n \partial_{i_1}\ldots$
$\partial_{i_n}Y_{\ell m}(\verr{})$ and
$(|\vec{r}\,|\partial_{i_1})\ldots(|\vec{r}\,|\partial_{i_n})Y_{\ell 
  m}(\verr{})$ do not depend on $|\vec{r}\,|$. We are thus led to
define the operator
\begin{equation}
  \label{eq:radial.trick}
\O_{(n)i_1\ldots i_n}=\left(-(n-1) \verr{i_n}+
  |\vec{r}\,|\partial_{i_n}\right) \left(-(n-2) \verr{i_{n-1}}+
  |\vec{r}\,|\partial_{i_{n-1}}\right) \ldots \left(-\verr{i_2}+
  |\vec{r}\,|\partial_{i_2}\right)
\left(|\vec{r}\,|\partial_{i_1}\right),
\end{equation}
satisfying the relation 
\begin{equation}
  \label{eq:radial.trick.2}
|\vec{r}\,|^n
\partial_{i_1}\ldots \partial_{i_n} f(\vec{r}\,)=
\O_{(n)i_1\ldots i_n} f(\vec{r}\,),
\end{equation}
which is straightforward to prove by induction.  If we set 
$f(\vec{r}\,)=Y_{\ell m}(\verr{})$ in (\ref{eq:radial.trick.2}), then
both sides are independent of $|\vec{r}\,|$.  We can therefore write
\begin{equation}
  \label{eq:nder.1}
|\vec{r}|^n \partial_{i_1}\ldots\partial_{i_n} Y_{\ell m}(\verr{}) =
\langle \verr{}|\O_{(n)i_1\ldots i_n}|\ell,m\rangle
=\sum_{\ell',m'} Y_{\ell'm'}(\verr{}) 
\langle\ell',m'|\O_{(n)i_1\ldots i_n}|\ell,m\rangle.
\end{equation}
Since $\O_{(n)i_1\ldots i_n}$ is totally symmetric by
(\ref{eq:radial.trick.2}), we can apply the theorem
(\ref{eq:tot.sym.we}) to the right--hand side of (\ref{eq:nder.1}).
By using (\ref{eq:tot.sym.we}) and (\ref{eq:tshy}), from
(\ref{eq:nder.1}) we get
\begin{equation}
  \label{eq:nder.2}
  \begin{split}
|\vec{r}|^n \partial_{i_1}\ldots\partial_{i_n} Y_{\ell m}(\verr{}) =
\sum_{\ell'=\ell-n}^{\ell+n}   
\sum_{\substack{s=0\\(n-s)\,\mathrm{even}}}^n
\sqrt{\frac{2\ell'+1}{2\ell+1}} (-1)^{\ell'-\ell}
\frac{\lambda_{(n,s)}^2}{n!} 
\langle\ell'||\veps{(s)}\cdot \mathrm{Tr}_{(s)}(\O_{(n)})||\ell\rangle 
\\\times
\left(Y^{\ell's}_{\ell m}(\verr{})
\right)_{\{i_1\ldots i_s} 
\delta_{i_{s+1}i_{s+2}}\ldots \delta_{i_{n-1}i_n\}},
  \end{split}
\end{equation}
where we assume $\ell\geq n$, $\lambda_{(n,s)}$ is defined in
(\ref{eq:tot.sym.def.b}) and $\mathrm{Tr}_{(s)}$in
(\ref{eq:tot.sym.expn}).  In order to make this expression more
explicit we need to evaluate the traces.  We do so by going back to
(\ref{eq:radial.trick.2}), with $f=Y_{\ell m}$, and using the fact
that $Y_{\ell m}$ is an eigenfunction of the Laplace operator, to
get\footnote{We   define double factorials for odd integers $n=2k-1$
  as $n!! = \prod_{h=1}^k (2h-1)$ and for even
  integers $n=2k$ as $n!! = \prod_{h=1}^k (2h)$ so that
  $n!=n!!(n-1)!!$.}
\begin{subequations}
  \label{eq:nder.34}
\begin{equation}
  \label{eq:nder.3}
|\vec{r}|^n \partial_{i_n}\ldots\partial_{i_{2k+1}}
\partial_{h_k}\partial_{h_k}\ldots \partial_{h_1}\partial_{h_1}
Y_{\ell m}(\verr{}) =    
(-1)^k \ell (\ell+1) \frac{(\ell-1)!!}{(\ell-2k+1)!!}
\frac{(\ell+2k-2)!!}{\ell!!} \O_{(n,2k+1)i_n\ldots i_{2k+1}}
Y_{\ell m}(\verr{}),
\end{equation}
where $\O_{(n,q)i_n\ldots i_{q}}$ is the rank-$(n-q+1)$ tensor operator
defined as
\begin{equation}
  \label{eq:nder.4}
\O_{(n,q)i_n\ldots i_{q}} = 
\left(-(n-1) \verr{i_n}+
  |\vec{r}\,|\partial_{i_n}\right) \left(-(n-2) \verr{i_{n-1}}+
  |\vec{r}\,|\partial_{i_{n-1}}\right) \ldots 
\left(-(q-1)\verr{i_q}+
  |\vec{r}\,|\partial_{i_q}\right).
\end{equation}
\end{subequations}
The operator $\O_{(n,q)}$ is a generalization of $\O_{(n)}$ defined in
(\ref{eq:radial.trick}), with $\O_{(n,1)} = \O_{(n)}$.  From
(\ref{eq:nder.34}) we obtain the sought--for traces of
$\O_{(n)}$ as
\begin{equation}
  \label{eq:nder.5}
\mathrm{Tr}_{(s)}(\O_{(n)})_{i_1\ldots i_s} = 
(-1)^{(n-s)/2} \ell (\ell+1) \frac{(\ell-1)!!}{(\ell-n+s+1)!!} 
\frac{(\ell+n-s-2)!!}{\ell!!} \O_{(n,n-s+1)i_1\ldots i_s},
\end{equation}
with $\mathrm{Tr}_{(s)}$ as defined in (\ref{eq:tot.sym.expn}). 
The reduced matrix elements of $\O_{(n,q)}$ appearing in
(\ref{eq:nder.2}) through the relation (\ref{eq:nder.5}) can be
evaluated with (\ref{eq:std.tens.upper.lower}) by means of the usual
recoupling techniques \cite{edm96,var88}.  We omit the algebra for
brevity and state here the result
\begin{equation}
  \label{eq:nder.6}
\langle \ell'||\veps{r}\cdot \O_{(n,q)}|| \ell\rangle =
(-1)^{(\ell'-\ell+r)/2} \sqrt{\frac{r!}{(2r-1)!!}}\sqrt{\frac{2\ell+1}
{2\ell'+1}} \frac{(\ell-q+2)!!}{(\ell'-n+1)!!} 
\frac{(\ell'+n-2)!!}{(\ell+q-3)!!} \CG{\ell}{0}{r}{0}{\ell'}{0},
\end{equation}
with $r=n-q+1$ the rank of $\O_{(n,q)}$.  Notice that the
CG coefficient on the right--hand side vanishes unless
$\ell+r-\ell'$ is even.

Putting together (\ref{eq:nder.2}), (\ref{eq:nder.5}) and
(\ref{eq:nder.6}) finally leads to
\begin{equation}
  \label{eq:nder.7}
  \begin{gathered}
|\vec{r}|^n \partial_{i_1}\ldots\partial_{i_n} Y_{\ell m}(\verr{}) =
\sum_{\ell'=|\ell-n|}^{\ell+n} 
\sum_{\substack{s=0\\(n-s)\,\mathrm{even}}}^n
\Theta_{ns}(\ell,\ell') \CG{\ell}{0}{s}{0}{\ell'}{0}
\left(Y^{\ell's}_{\ell m}(\verr{})
\right)_{\{i_1\ldots i_s} 
\delta_{i_{s+1}i_{s+2}}\ldots \delta_{i_{n-1}i_n\}},
\\
\Theta_{ns}(\ell,\ell') = (-1)^{\frac{\ell-\ell'+s}{2}} (-1)^{n_\delta}
\frac{2s+1}{2^{n_\delta}n_\delta! (n+s+1)!!}
\sqrt{\frac{(2s-1)!!}{s!}} \frac{(\ell+1)!!}{(\ell-2)!!} 
\frac{(\ell'+n-2)!!}{(\ell'-n+1)!!},\quad n_\delta=\frac{n-s}{2}.
  \end{gathered}
\end{equation}
To the best of our knowledge, this general result has not been given
in the previous literature. It is certainly not to be found in the
references listed in the bibliography.

Equation (\ref{eq:nder.7}) gives an explicit expression for the
$n^\mathrm{th}$ derivatives of $Y_{\ell m}$ in terms of tensor
spherical harmonics.  For each $s$, $\ell$, $m$, equation
(\ref{eq:nder.7}) can be inverted to give $Y^{\ell's}_{\ell
  m}(\verr{})$ as a linear combination of derivatives of $Y_{\ell
  m}(\verr{})$.  Those relations can, in fact, be taken as the
definition of tensor spherical harmonics.  For instance, in the case
$n=1$ from (\ref{eq:nder.7}) and (\ref{eq:NMC}) we get
\begin{equation}
\label{eq:nder.8}
    \begin{gathered}
Y^{(\ell-1)1}_{\ell m}(\verr{}) = \frac{|\vec{r}\,|^{-(\ell-1)}}{\sqrt{\ell(2\ell+1)}} 
\nabla\left(|\vec{r}\,|^\ell Y_{\ell m}(\verr{})\right),
\qquad
Y^{\ell1}_{\ell m}(\verr{}) = -\frac{i}{\sqrt{\ell(\ell+1)}}
\vec{r}\wedge\nabla Y_{\ell m}(\verr{}),\\
Y^{(\ell+1)1}_{\ell m}(\verr{}) = \frac{|\vec{r}\,|^{\ell+2}}{\sqrt{(\ell+1)(2\ell+1)}} 
\nabla\left(
|\vec{r}\,|^{-(\ell+1)} Y_{\ell m}(\verr{})\right).
    \end{gathered}
\end{equation}
which allows us to make contact with vector spherical harmonics as
defined in electromagnetism \cite{jac99,hil54}.

We mention, finally, that from (\ref{eq:nder.7}) we can obtain an
expression for the matrix elements of momentum operators $\langle
\ell',m'| p_{i_1}\ldots p_{i_n}|\ell,m\rangle$.  Indeed, multiplying
(\ref{eq:nder.7}) by $Y_{\ell'm'}(\verr{})^*$ and integrating over the
unit sphere, with the help of (\ref{eq:tshy}) we get
\begin{equation}
  \label{eq:matpn}
  \begin{aligned}
  \int\! d^2\verr{}\;
  Y_{\ell'm'}(\verr{})^* \partial_{i_1}\ldots\partial_{i_n} Y_{\ell
    m}(\verr{}) = |\vec{r}|^{-n}
\sum_{\substack{s=0\\(n-s)\,\mathrm{even}}}^n
\Theta_{ns}(\ell,\ell') \frac{n!}{\lambda_{(n,s)}}
&\CG{\ell}{0}{s}{0}{\ell'}{0}
\CG{\ell'}{m'}{s}{m-m'}{\ell}{m}\\
&\times\veps{\{n,s\}i_1\ldots i_n}(m-m'),
  \end{aligned}
\end{equation}
with $\Theta_{ns}$ as defined in (\ref{eq:nder.7}) and
$\lambda_{(n,s)}$ in (\ref{eq:tot.sym.def.b}). 

\section{The Wigner-Eckart theorem for partially irreducible tensors}
\label{sec:parirr}

The last class of reducible tensors we consider is that of partially
irreducible tensors, defined as those tensors $T_{i_1\ldots i_{n+1}}$
of rank $n+1$ totally symmetric and traceless in their first $n$
indices.  They play a r\^ole in the cartesian multipole expansion of
the electromagnetic vector potential discussed in section
\ref{sec:mmultipole}.  A rank-$n+1$ partially irreducible tensor
possesses $6n+3$ independent components: $2n+3$ corresponding to
the irreducible part of spin $s=n+1$, $2n+1$ to the parts
antisymmetric in some pair $i_ki_{n+1}$ ($1\leq k\leq n$) of spin
$s=n$, and $2n-1$ to the trace parts of spin $s=n-1$.  A basis of the
space of rank-$n+1$ partially irreducible tensor is then given by
\begin{equation}
  \label{eq:partir1}
  B^{(j)}_{i_1\ldots i_{n+1}}(m) =
  \sum_{\mu,\nu} \CG{n}{\mu}{1}{\nu}{j}{m} \veps{(n)i_1\ldots
    i_{n}}(\mu)\veps{(1)i_{n+1}}(\nu),
\qquad j=n-1,n,n+1,
\qquad -j \leq m \leq j.
\end{equation}
From (\ref{eq:std.tens.s.prop}) and standard properties of the
Clebsch-Gordan coefficients the basis tensors (\ref{eq:partir1}) are
found to satisfy the orthonormality and complex conjugation relations
\begin{equation}
  \label{eq:partir2}
  B^{(j)}_{i_1\ldots i_{n+1}}(m)^* B^{(j')}_{i_1\ldots i_{n+1}}(m') =
  \delta_{jj'}\delta_{mm'}, 
\qquad
B^{(j)}_{i_1\ldots i_{n+1}}(m)^* = (-1)^{n+1-j} (-1)^m
B^{(j)}_{i_1\ldots i_{n+1}}(-m). 
\end{equation}
In order to obtain matrix elements of partially irreducible tensor
operators we need to express the reducible rank-$n+1$ basis tensors 
(\ref{eq:partir1}) as (linear combinations of) direct products of an
irreducible tensor of rank $s$, with $s=n+1,$ $n$, $n-1$, with a
rank-($n+1-s$) spin-0 (i.e., isotropic) tensor:
\begin{subequations}
  \label{eq:partir3}
\begin{align}
  \label{eq:partir3a}
B^{(n+1)}_{i_1\ldots i_{n+1}}(m)  &= \veps{(n+1)i_1\ldots i_{n+1}},\\
  \label{eq:partir3b}
B^{(n)}_{i_1\ldots i_{n+1}}(m) &= \frac{i}{\sqrt{n(n+1)}} \sum_{k=1}^n
\varepsilon_{i_ki_{n+1}h} \veps{(n)hi_1\ldots \widehat{i}_{k}\ldots
  i_n},\\
  \label{eq:partir3c}
B^{(n-1)}_{i_1\ldots i_{n+1}}(m) &= \frac{1}{n\sqrt{(2n-1)(2n+1)}}
\left( 2
\sum_{1\leq k<h\leq n}\veps{(n-1)i_1\ldots\widehat{i}_k\ldots
  \widehat{i}_h \ldots i_{n+1}}\delta_{i_ki_h} \right.\\
&\quad - (2n-1)\left.
\sum_{k=1}^n \veps{(n-1)i_1\ldots\widehat{i}_k\ldots i_n} \delta_{i_ki_{n+1}}
\right),\nonumber
\end{align}
\end{subequations}
where the caret over a subscript indicates that it is to be
omitted. Equation (\ref{eq:partir3a}) is just (\ref{eq:std.tens.s});
proofs of (\ref{eq:partir3b}), (\ref{eq:partir3c}) as direct
consequences of (\ref{eq:nder.8}) are given in Appendix
\ref{sec:appD}.

Given a partially irreducible tensor $T_{i_1\ldots i_{n+1}}$, from
(\ref{eq:partir1})-(\ref{eq:partir3}) we have,
\begin{equation}
  \label{eq:partir4} 
T_{i_1\ldots i_{n+1}} = \sum_{j=n-1}^{n+1}\sum_{m=-j}^j T^{(j)}(m)
B^{(j)}_{i_1\ldots i_{n+1}}(m)^*,
\end{equation}
with
\begin{subequations}
  \label{eq:partir5} 
\begin{align}
  \label{eq:partir5a}
T^{(n+1)}(m) &= B^{(n+1)}_{i_1\ldots i_{n+1}}(m) T_{i_1\ldots i_{n+1}} 
  = \veps{(n+1)i_1\ldots i_{n+1}}(m) T_{i_1\ldots i_{n+1}},\\
  \label{eq:partir5b}
T^{(n)}(m) &= B^{(n)}_{i_1\ldots i_{n+1}}(m) T_{i_1\ldots i_{n+1}} 
  = i\sqrt{\frac{n}{n+1}} \veps{(n)i_1\ldots
    i_{n-1}h}(m)\varepsilon_{hi_ni_{n+1}} T_{i_1\ldots i_{n+1}},\\
  \label{eq:partir5c}
T^{(n-1)}(m) &= B^{(n-1)}_{i_1\ldots i_{n+1}}(m) T_{i_1\ldots i_{n+1}} 
  = -\sqrt{\frac{2n-1}{2n+1}}\veps{(n-1)i_1\ldots
    i_{n-1}}(m) T_{i_1\ldots i_{n-1}jj}.
\end{align}
\end{subequations}
Notice that, due to the symmetry properties of $T_{i_1\ldots
  i_{n+1}}$, the tensor contraction on the r.h.s.\ of
(\ref{eq:partir5b}) can also be written as $\veps{(n)i_1\ldots
  i_{k-1}hi_{k+1}\ldots i_{n}}(m)\varepsilon_{hi_ki_{n+1}}
T_{i_1\ldots i_{n+1}}$ with $1\leq k \leq n$.  Similarly, the tensor
contraction on the r.h.s.\ of (\ref{eq:partir5c}) can be written as 
$\veps{(n-1)i_1\ldots\widehat{i}_k\ldots i_{n}}(m) T_{i_1\ldots
  i_{k-1}ji_{k+1}\ldots i_nj}$ with $1\leq k \leq n$. The forms used
in (\ref{eq:partir5}) were chosen for notational convenience.

The
components $T^{(j)}(m)$ defined in (\ref{eq:partir5}) are \sitos, by
Corollary \ref{crl:sito.from.cito}, so from (\ref{eq:partir4}) and the
Wigner-Eckart theorem (\ref{eq:WE.sito}) we
obtain the equality
\begin{equation}
  \label{eq:partir6}
\langle j',m'|T_{i_1\ldots i_{n+1}}|j,m\rangle =
\sum_{s=n-1}^{n+1} \langle j'||T^{(s)}||j\rangle
\sum_{s_z=-s}^s \CG{j}{m}{s}{s_z}{j'}{m'} B^{(s)}_{i_1\ldots
  i_{n+1}}(s_z)^*, 
\end{equation}
which is the Wigner-Eckart theorem for partially irreducible cartesian
tensor operators.

\subsection{Magnetic multipole moments}
\label{sec:mmultipole}

In this section we derive the spherical and cartesian multipole
expansions for the magnetic vector potential, and the relation between
them to all orders.  We restrict ourselves here to the case where
time-retardation effects can be neglected.  In that case 
the magnetic potential is given by \cite{jac99} 
\begin{equation}
  \label{eq:mag1}
  A_k(\vec{r},t) = \frac{\mu_0}{4\pi}
\int_V d^3r' \frac{j_k(\vec{r}\,',t)}{|\vec{r} - \vec{r}\,'|}
=\frac{\mu_0}{4\pi} \int_Vd^3r'\sum_{\ell=0}^\infty
\frac{4\pi}{2\ell+1}\frac{|\vec{r}\,'|^\ell}{|\vec{r}|^{\ell+1}}
\sum_{\mu=-\ell}^\ell Y_{\ell\mu}(\verr{}\,')^* Y_{\ell\mu}(\verr{}) j_k(\vec{r}\,',t),
\end{equation}
where in the last equality we substituted the expansion of $1/|\vec{r}
- \vec{r}\,'|$ in spherical harmonics \cite{jac99}.  In
(\ref{eq:mag1}) we can rewrite
\begin{subequations}
\label{eq:mag2}
\begin{equation}
\label{eq:mag2a}
\sum_{\mu=-\ell}^\ell Y_{\ell\mu}(\verr{}\,')^* Y_{\ell\mu}(\verr{})
j_k(\vec{r}\,',t)
=\sum_{\mu=-\ell}^\ell \sum_{\mu'=-\ell}^\ell Y_{\ell\mu'}(\verr{}\,')^* Y_{\ell\mu}(\verr{}) j_i(\vec{r}\,',t)  \delta_{ik}\delta_{\mu\mu'},
\end{equation}
with
\begin{equation}
\label{eq:mag2b}
  \begin{aligned}
\delta_{ik}\delta_{\mu\mu'} &= \sum_{\nu=-1}^1
\veps{(1)i}(\nu)^* \veps{(1)k}(\nu) \delta_{\mu\mu'}
= \sum_{\nu=-1}^1 \sum_{\nu'=-1}^1\veps{(1)i}(\nu')^* \veps{(1)k}(\nu) 
\delta_{\mu\mu'} \delta_{\nu\nu'}\\
&=
\sum_{j=|\ell-1|}^{\ell+1}\sum_{m=-j}^j\sum_{\nu,\nu'=-1}^1 
\veps{(1)i}(\nu')^* \veps{(1)k}(\nu) 
\CG{\ell}{\mu'}{1}{\nu'}{j}{m}\CG{\ell}{\mu}{1}{\nu}{j}{m}.
  \end{aligned}
\end{equation}
\end{subequations}
By inserting (\ref{eq:mag2b}) in (\ref{eq:mag2a}) and the result in
(\ref{eq:mag1}) we obtain
\begin{equation}
  \label{eq:mag3}
\vec{A}(\vec{r},t) = \frac{\mu_0}{4\pi} \sum_{\ell=0}^\infty 
\frac{4\pi}{2\ell+1}\frac{1}{|\vec{r}|^{\ell+1}} 
\sum_{j=|\ell-1|}^{\ell+1}\sum_{m=-j}^j \mmp{\ell}{j}{m}(t)^* Y^{\ell
  1}_{jm}(\verr{}), 
\quad
\mmp{\ell}{j}{m}(t) = \int_Vd^3r'|\vec{r}\,'|^{\ell}
\vj(\vec{r}\,',t)\cdot Y^{\ell 1}_{jm}(\verr{}\,').
\end{equation}
This expression is a spherical multipole expansion, but it must be
simplified since it contains redundant terms. We consider first the
term in (\ref{eq:mag3}) with $j=\ell-1$ ($\ell\geq1$).  From the third
equality in (\ref{eq:nder.8}) this term is seen to be a gradient that
does not contribute to the magnetic field.  Furthermore, since the
argument of the gradient is a harmonic function, it is divergenceless.
This term plays a r\^ole in setting boundary conditions such as
$\vec{A}(R,\theta,\varphi) = \vec{v}(\theta,\varphi)$ or
$\verr{}\cdot\vec{A}(R,\theta,\varphi) = 0$ for some finite value
$r=R$.  We restrict ourselves here to the case $\vec{A}\rightarrow0$
as $r\rightarrow\infty$, so we will set $\mmp{\ell}{(\ell-1)}{m}=0$
from here on.

The term with $j=\ell+1$ ($\ell\ge0$) in (\ref{eq:mag3}) is related to
the time derivative of the electric multipole moment $q_{(\ell+1)m}$
given in (\ref{eq:sph.multipole}).  To show this we adopt the method
used in \cite{kie98}.  We consider a closed surface $S$ with exterior
normal $\ver{n}$ containing the volume $V$ and its boundary in its
interior, so that $\vj=0$ on $S$.  By means of Gauss' divergence
theorem and the continuity equation $\partial_k \vj_k + \dot{\rho}=0$
we obtain 
\begin{equation*}
0 = \oint_{S}d^2r' \ver{n}_k \vj_k(\verr{}\,') r'_{i_1} \ldots
r'_{i_n} = \int_V d^3r' \left(
\vj_k(\verr{}\,') \partial'_k (r'_{i_1} \ldots r'_{i_n}) -
\dot{\rho}(\verr{}\,') r'_{i_1} \ldots r'_{i_n}
\right).
\end{equation*}
Multiplication of this equation by $\veps{(n)i_1\cdots i_n}(m)$ and
use of (\ref{eq:shyfinal}) leads to
\begin{equation}
  \label{eq:mag4}
\int_V d^3r' \vj_k(\verr{}\,') \partial'_k \left( |\vec{r}\,'|^n
  Y_{nm}(\verr{}\,')\right)   = \dot{q}_{nm},
\end{equation}
with $q_{nm}$ the electric multipole moments defined in
(\ref{eq:sph.multipole}).  The gradient in the integrand in
(\ref{eq:mag4}) is given by the first equality in (\ref{eq:nder.8}),
which yields 
\begin{equation}
  \label{eq:mag6}
  \sqrt{(\ell+1)(2\ell+3)} \mmp{\ell}{(\ell+1)}{m} =
  \dot{q}_{(\ell+1)m}. 
\end{equation}
We then have
\begin{equation}
  \label{eq:mag7}
  \begin{aligned}
    \vec{A}(\vec{r},t) &= \frac{\mu_0}{4\pi} \sum_{\ell=1}^\infty 
    \frac{4\pi}{2\ell+1} \frac{1}{|\vec{r}|^{\ell+1}}
    \sum_{m=-\ell}^\ell \mmp{\ell}{\ell}{m}(t)^* Y^{\ell 1}_{\ell
      m}(\verr{}) \\
&+\frac{\mu_0}{4\pi} \sum_{\ell=0}^\infty 
    \frac{4\pi}{2\ell+1} \frac{1}{\sqrt{(\ell+1)(2\ell+3)}}
    \frac{1}{|\vec{r}|^{\ell+1}}
    \sum_{m=-(\ell+1)}^{\ell+1} \dot{q}_{(\ell+1)m}(t)^* Y^{\ell 1}_{(\ell+1)
      m}(\verr{}),
  \end{aligned}
\end{equation}
with $\mmp{\ell}{\ell}{m}$, $q_{(\ell+1)m}$ as defined in
(\ref{eq:mag3}), (\ref{eq:sph.multipole}), respectively.  Equation
(\ref{eq:mag7}) is the spherical multipole expansion of the vector
potential.  We notice that in the magnetostatic case $\partial_k j_k =
0 = \dot{\rho}$, only the first line in (\ref{eq:mag7}) is
non-vanishing.  In that case (\ref{eq:mag7}) yields a vector potential
that is transverse both in momentum and coordinate space,
$\nabla\cdot\vec{A}=0=\vec{r}\cdot\vec{A}$.

In order to obtain the cartesian multipole expansion of the magnetic
potential, and its relation to the spherical one, we could proceed as
in the case of the electric potential, with (\ref{eq:legendre1}) as a
starting point, to obtain \cite{kie98}
\begin{equation*}
  A_i(\vec{r}) = \frac{\mu_0}{4\pi} \sum_{\ell=0}^\infty
  \frac{1}{\ell!} \frac{1}{|\vec{r}\,|^{2\ell+1}}
  \widetilde{M}_{i_1\ldots i_\ell i} r_{i_1}\ldots r_{i_\ell},
\quad
\widetilde{M}_{i_1\ldots i_\ell i} = (2\ell-1)!! \int_V\!d^3r'\, 
r'_{(i_1}\ldots r'_{i_\ell)_0} j_i(\vec{r}\,').
\end{equation*}
This cartesian multipole expansion is given in terms of the partially
irreducible magnetic multipole tensor $\widetilde{M}_{i_1\ldots i_\ell
  i}$, thus, it may further unraveled into irreducible components by
means of the decomposition (\ref{eq:partir4}), (\ref{eq:partir5}) for
$\widetilde{M}$ \cite{kie98}. Equivalently, we will derive the
cartesian expansion from the spherical one (\ref{eq:mag7}) to take
advantage of the simplifications already carried out on it.  

We consider first the terms on the first line of (\ref{eq:mag7}).  The
vector spherical harmonic $Y^{\ell 1}_{\ell m}(\verr{})$, as defined
by (\ref{eq:tshy}), can be expanded in tensorial form by multiplying
(\ref{eq:partir3b}) by $\verr{i_1}\ldots\verr{i_n}$, as done below in
(\ref{eq:appD3}). Inserting that tensorial expansion in the first line
of (\ref{eq:mag7}), after some rearrangements we get,
\begin{equation}
  \label{eq:mag9}
  \begin{aligned}
A^{(1)}_j (\vec{r},t) &= \frac{\mu_0}{4\pi} \sum_{\ell=1}^\infty
\frac{1}{\ell!} \frac{1}{|\vec{r}|^{2\ell+1}} r_{k_1}\ldots r_{k_\ell}
\varepsilon_{jk_\ell h} M_{hk_1\ldots k_{\ell-1}}(t)^*,\\
M_{k_1\ldots k_{\ell}}(t) &= \sum_{m=-\ell}^\ell i
\frac{4\pi}{2\ell+1} \ell! \sqrt{\frac{\ell}{\ell+1}} \mmp{\ell}{\ell}
  {m}(t) \veps{(\ell)k_1\ldots k_{\ell}}(m).
  \end{aligned}
\end{equation}
This equality gives the cartesian multipole expansion of the first
line of (\ref{eq:mag7}). $M_{k_1\ldots k_{\ell}}(t)$ is the magnetic
$2^\ell$-multipole irreducible tensor, given in (\ref{eq:mag9}) in
terms of the spherical one $\mmp{\ell}{\ell}{m}(t)$.  Multiplication
of the second line of (\ref{eq:mag9}) by $\veps{(\ell)hk_1\cdots
  k_{\ell-1}}(m)$ yields the inverse relation:
\begin{equation}
  \label{eq:maga}
  \mmp{\ell}{\ell}{m}(t) = -i \frac{2\ell+1}{4\pi} \frac{1}{\ell!}
  \sqrt{\frac{\ell+1}{\ell}} \veps{(\ell)k_1\cdots k_{\ell}}
  M_{k_1\ldots k_{\ell}}(t).
\end{equation}
Furthermore, by replacing in the second line of (\ref{eq:mag9}) the
definition (\ref{eq:mag3}) of $\mmp{\ell}{\ell}{m}$ and using
(\ref{eq:appD3}) again, we have 
\begin{equation}
  \label{eq:magb}
M_{k_1\ldots k_{\ell}}(t) = \frac{4\pi}{2\ell+1} \ell!
\frac{\ell}{\ell+1} \sum_{m=-\ell}^\ell \veps{(\ell)k_1\cdots
  k_{\ell}}(m)^* \veps{(\ell)h'k'_1\cdots  k'_{\ell-1}}(m) 
\int_V d^3r'\left(\vj(\vec{r}\,',t)\wedge\vec{r}\,'\right)_{h'}
\vec{r}\,'_{k'_1} \cdots \vec{r}\,'_{k'_{\ell-1}}. 
\end{equation}
This equation gives the irreducible magnetic multipole tensor
$M_{k_1\ldots k_{\ell}}$, up to a normalization, as the irreducible
component of the integral on the right-hand side.

We turn next to the second line in (\ref{eq:mag7}).  Since $Y^{\ell 1}_{(\ell+1)
      m}(\verr{})$ is maximally coupled, we can use
    (\ref{eq:std.tens.s}) (or, equivalently,
    (\ref{eq:maximal.coupling})) and (\ref{eq:shyfinal}) in
    (\ref{eq:tshy}) to write
    \begin{equation}
      \label{eq:magc}
\left(Y^{\ell 1}_{(\ell+1)m}\right) = N_\ell \veps{(\ell+1)j_1\ldots
  j_\ell j}(m) \verr{j_1}\ldots\verr{j_\ell}.
    \end{equation}
With this expression, and equation (\ref{eq:multipole.final})
expressing the spherical electric multipoles $q_{\ell m}$ in terms of
the cartesian ones $Q_{i_1\ldots i_\ell}$, the second line of
(\ref{eq:mag7}) can be written as
\begin{equation}
  \label{eq:magd}
A^{(2)}_j (\vec{r},t) = \frac{\mu_0}{4\pi} \sum_{\ell=1}^\infty
\frac{1}{\ell!} \frac{1}{|\vec{r}|^{2\ell+1}} r_{k_1}\ldots r_{k_\ell}
\frac{1}{(\ell+1)(2\ell+1)} \dot{Q}_{j_1\dots j_\ell j},
\end{equation}
where we have used the fact that $Q_{j_1\dots j_\ell j}$ is a
rank-$(\ell+1)$ irreducible tensor.  From (\ref{eq:mag7}),
(\ref{eq:mag9}), (\ref{eq:magd}) we get the complete, fully reduced
magnetic cartesian multipole expansion 
\begin{equation}
  \label{eq:mage}
A_j(\vec{r},t) = \frac{\mu_0}{4\pi} \frac{1}{|\vec{r}\,|} \dot{Q}_j(t)^* +
\frac{\mu_0}{4\pi} \sum_{\ell=1}^\infty
\frac{1}{\ell!}\frac{1}{|\vec{r}\,|^{2\ell+1}} r_{k_1}\ldots
r_{k_\ell} \left(
\varepsilon_{jk_\ell h}M_{hk_1\ldots k_{\ell-1}}(t)^*
+\frac{1}{(\ell+1)(2\ell+1)}\dot{Q}_{k_1\dots k_\ell j}(t)^*
\right),
\end{equation}
where we have separated the time derivative of the electric dipole
moment which is the only $O(1/r)$ term.  The expansion (\ref{eq:mage})
has been previously given in \cite{kie98}.  Here we recover that
result and extend it by giving the relations (\ref{eq:mag9}),
(\ref{eq:maga}) between the spherical magnetic moments
(\ref{eq:mag3}), (\ref{eq:mag7}) and the cartesian ones
(\ref{eq:magb}), (\ref{eq:mage}) to all orders in the multipole
expansion. Equations (\ref{eq:mag9}), (\ref{eq:maga}) give an explicit
form for the analogous relations found in \cite{tor02}. 

\section{Final remarks}
\label{sec:finrem}

In this paper we have established in full generality the explicit form
of the unitary correspondence between spherical and cartesian
irreducible tensor operators relative to a generic angular-momentum
operator $\vec{J}$.  The matrix elements of that unitary
transformation are the standard rank-$s$ tensors $\veps{(s)}(m)$
defined in section \ref{sec:spin.wf}, which constitute an
orthonormal basis for the space of rank-$s$ irreducible tensors, and
also a standard basis of spin-$s$ wave functions.  We defined
$\veps{(s)}(m)$ in three equivalent ways, namely, recursively
(\ref{eq:std.tens.s}), explicitly (\ref{eq:std.tens.s.alt}) and
implicitly (\ref{eq:std.tens.s.impl}). The tensors $\veps{(s)}(m)$
satisfy all of the usual phase and coupling conventions required to
make them \emph{bona-fide} angular-momentum wave functions.
Furthermore, they transform under rotations equivalently as cartesian
or spherical tensors.  We remark here that an analogous basis of
cartesian irreducible spinors can be constructed along the same lines
\cite{bou1,zem65,rar41}.  Both the tensor and spinor bases are of the
(non-relativistic) Rarita-Schwinger \cite{rar41} type.  With the spin
wave functions $\veps{(s)}(m)$ as matrix elements of a unitary change of
basis, we establish the relation between cartesian and spherical
irreducible tensors of any rank in section \ref{sec:csitos}.

The unitary mapping described in \ref{sec:csitos} is important because
it allows us to apply the methods and technical tools of quantum
angular--momentum theory to tensor algebra and vice versa.  The
interrelation of tensor and angular--momentum methods is used in
section \ref{sec:we} to extend the Wigner--Eckart theorem to cartesian
irreducible tensor operators of any rank, thus determining the form of
their matrix elements.  In principle, this result can be applied to
any cartesian tensor operator, by linearly decomposing it into its
irreducible components. In section \ref{sec:totsym} we take a step
towards that goal by extending the Wigner--Eckart theorem to totally
symmetric reducible tensors, which is the main result of this
paper. Such an extension allows us to obtain matrix elements for
arbitrary tensor powers of the position or momentum operators.  This
is exploited in section \ref{sec:dershy} to give an explicit
expression for the gradients of any order of spherical harmonics.  An
additional extension of the Wigner-Eckart theorem to partially
irreducible cartesian tensor operators is given in section
\ref{sec:parirr}.  Partially irreducible tensors occur in some
applications, such as the magnetic multipole expansion discussed in
\ref{sec:mmultipole}. 

On the other hand, the results of section \ref{sec:spin.wf} and
\ref{sec:csitos} also allow us to find the cartesian form of standard
spherical tensors.  Such cartesian tensorial forms provide a
complementary approach to the usual analytic methods often based on
differential equations.  In section \ref{sec:finrot} we show that
Wigner $D$-matrices are the spherical components of three-dimensional
rotation matrices, which leads to an expression of $D^\ell$ as a
linear combination of products of $D^1$ matrices. In section
\ref{sec:standard} we obtain the cartesian tensorial form of ordinary,
bipolar and tensor spherical harmonics, and spin-polarization
operators.  We find some relations between those standard
functions that are of interest, like the decomposition of $D$-matrices
mentioned above, the binomial expansion for spherical harmonics, and
the relation between tensor and bipolar spherical harmonics, based on
the tensorial representations for them.  We discuss also the relation
between spherical and cartesian multipoles of any rank, both for the
scalar electric (section \ref{sec:emultipole}) and vector magnetic
(\ref{sec:mmultipole}) potentials. The best illustration of the
kind of relations referred to here, and of the power of our approach,
is the explicit expression for gradients of any order of spherical
harmonics found in section \ref{sec:dershy}.

\section*{Acknowledgements}

The author has been partially supported by Sistema Nacional de
Investigadores de M\'exico.

\appendix
\renewcommand{\theequation}{\thesection.\arabic{equation}}

\setcounter{equation}{0}
\section{Angular momentum: notation and conventions}
\label{sec:angmom}

In this appendix we state our conventions for quantum
angular--momentum theory
\cite{wig59,lau77,mes61,coh77,gal90,bie09,var88,con57,edm96}.
Throughout this paper we follow the convention of \cite{lau77} that
angular--momentum operators are dimensionless.  In particular, the
orbital angular--momentum operator for a single particle is
$\vec{L}=(1/\hbar) \vec{r} \wedge \vec{p} = -i \vec{r} \wedge
\vec{\nabla}$.  In order to switch to the alternate convention in
which angular--momentum operators have units of $\hbar$ it suffices to
replace $\vec{J}$ (resp.\ $\vec{L}$, $\vec{S}$) in our equations by
$\vec{J}/\hbar$ (resp.\ $\vec{L}/\hbar$, $\vec{S}/\hbar$).

Let $\vec{J}$ be a generic angular-momentum operator, $[J_h,J_i]=i
\eps_{hik}J_k$.  The simultaneous eigenstates of $\vec{J}\,^2$ and $J_3$
are denoted $|j,m\rangle$, with $j\geq 0$ and $-j\leq m\leq j$ integer
or half-odd-integer.  The state vectors $|j,m\rangle$ are assumed to be
normalized, therefore orthonormal, and their relative phases are
chosen so that they satisfy the usual convention
\begin{subequations}
\label{eq:matelmj}
\begin{equation}
  \label{eq:ladderj}
  \langle j,m'|J_{\pm} | j, m\rangle = \sqrt{(j\mp m)(j\pm m+1)}
  \delta_{m'(m\pm 1)},
\quad
  \langle j,m'|J_3 | j, m\rangle = m  \delta_{m'(m\pm 1)}.
\end{equation}
In particular, they satisfy the Condon--Shortley phase convention
$\langle j,m'|J_{\pm} | j, m\rangle\geq0$ \cite{con57,edm96,gal90}.
For the cartesian components of $\vec{J}$ we then have
\begin{equation}
  \label{eq:cartj}
  \langle j,m'|J_{\substack{1\\2}} | j, m\rangle = \frac{1}{2} \sqrt{j(j+1)-m m'}
  (\delta_{m'(m+1)}\pm\delta_{m'(m-1)}),  
\end{equation}
\end{subequations}
and the matrix element of $J_3$ as in (\ref{eq:ladderj}).  By using
the relation 
\begin{equation}
  \label{eq:cg.particular}
\CG{j}{m}{1}{\epsilon}{j}{m'} = \frac{1}{\sqrt{j(j+1)}}\times
\begin{cases}
  \frac{\displaystyle -\epsilon}{\displaystyle\sqrt{2}}\sqrt{(j-\epsilon
    m)(j+\epsilon m+1)}\,\delta_{m'(m+\epsilon)} & \mathrm{if} 
  \quad\epsilon=\pm1 \\
  m\,\delta_{mm'} & \mathrm{if}\quad\epsilon =0
\end{cases},
\end{equation}
equation (\ref{eq:ladderj}) can be written more compactly as
\begin{subequations}
\label{eq:matelmj2}
\begin{equation}
  \label{eq:spherj}
  \langle j,m'|\veps{(1)}(\epsilon)\cdot\vec{J} | j, m\rangle = 
\sqrt{j(j+1)} \CG{j}{m}{1}{\epsilon}{j}{m'},
\qquad
\epsilon=0,\pm1,
\end{equation}
with $\veps{(1)i}(\epsilon)$ defined in (\ref{eq:stdvec}), from whence we
obtain 
\begin{equation}
  \label{eq:cartj2}
  \langle j,m'|J_k | j, m\rangle = 
\sqrt{j(j+1)} \CG{j}{m}{1}{m'-m}{j}{m'}\,\veps{k}(m'-m)^*,
\end{equation}
which is equivalent to (\ref{eq:cartj}).
\end{subequations}

For integer $j$ the spatial wave--function associated with the orbital
angular--momentum eigenstates is given by the spherical harmonics, for
which we follow the modern conventions \cite{coh77,gal90,var88,jac99}
(a different phase convention for $Y_{\ell m}$ is used in older texts 
\cite{lau77,con57,edm96})
\begin{equation}
  \label{eq:shy.std}
  Y_{\ell m}(\verr{}) = \sqrt{\frac{2\ell+1}{4\pi}
    \frac{(\ell-m)!}{(\ell+m)!}} P_{\ell m}(\cos\theta) e^{im\varphi}, 
\qquad
P_{\ell m}(x) = \frac{(-1)^{\ell+m}}{\ell!2^\ell}(1-x^2)^{m/2} 
\frac{d^{\ell+m}}{dx^{\ell+m}} (1-x^2)^\ell,
\end{equation}
where $\theta$, $\varphi$ are the polar coordinates of $\verr{}$, and
where $P_{\ell m}(x)$ is an associated Legendre function of the first
kind (as defined, e.g., in \cite{gal90,jac99}) that for $m=0$ reduces
to a Legendre polynomial $P_\ell(x)$.  An explicit expression for
$P_{\ell m}(x)$ is given in section \ref{sec:xplct}.

We turn next to the coupling of angular momenta.  Consider two
subsystems with angular-momentum operators $\vec{J}_a$, $a=1$, 2, each
one acting on the Hilbert space of states $\H_a$ of its subsystem.
Let $\H=\H_1\otimes\H_2$ be the state space of the total system, on
which the total angular-momentum operator $\vec{J}=\vec{J}_1\otimes
I_2+I_1\otimes\vec{J}_2$ acts, where $I_a$ is the identity operator on
$\H_a$.  The definition of $\vec{J}$ is usually written in the
abbreviated form $\vec{J}=\vec{J}_1+\vec{J}_2$.  If each subsystem $a$
is in a state $|j_a,\mu_a\rangle$, simultaneous eigenstate of
$(\vec{J}_a)^2$ and $(J_a)_3$, the state of the total system is the
tensor product state $|j_1,\mu_1;j_2,\mu_2 \rangle \equiv |j_1,
\mu_1\rangle \otimes |j_2, \mu_2\rangle$.  On the other hand, the
states with definite total angular momentum $|j_1, j_2, j, m \rangle$
are chosen to be simultaneous eigenstates of $(\vec{J}_1)^2$,
$(\vec{J}_2)^2$, $\vec{J}\,^2$ and $\vec{J}_3$.  The unitary
transformation relating the two orthonormal bases of $\H$ is given by
the CG coefficients, which we denote as
\begin{equation}
  \label{eq:CG.coupling}
|j_1, j_2, j, m \rangle = \sum_{\mu_1,\mu_2}
\CG{j_1}{\mu_1}{j_2}{\mu_2}{j}{m} |j_1,\mu_1;j_2,\mu_2 \rangle.
\end{equation}
A global complex factor in the CG coefficients is
determined by the normalization condition and the Condon--Shortley
\cite{con57} phase convention $\CG{j_1}{j_1}{j_2}{j-j_1}{j}{j}\geq 0$,
which implies in particular that all CG coefficients are
real.   
Since $|j_1, j_2,
j, m \rangle$ is by definition an eigenstate of $\vec{J}\,^2$ and
$\vec{J}_3$, the matrix elements of $\vec{J}$ between those states
do not depend on $j_1$, $j_2$, 
\begin{equation}
  \label{eq:prop.1}
\langle j_1, j_2,j', m'| \vec{J}\,| j_1, j_2,j, m \rangle=
\langle j', m'| \vec{J}\,| j, m \rangle=
\langle j, m' | \vec{J}\,| j, m \rangle \delta_{j'j},
\end{equation}
where the last matrix element is the standard one as given in
(\ref{eq:cartj}) or (\ref{eq:cartj2}).  

\paragraph*{Proof of (\ref{eq:spin.matrix.element})}

To prove (\ref{eq:spin.matrix.element}) we proceed by induction on
$n$.  The case $n=1$ is (\ref{eq:spin.mat.elm.1}), which is seen to be
true by exhaustive evaluation of all possible cases
$m',m=0,\pm1$. Assuming (\ref{eq:spin.matrix.element}) to be true, we
have to prove that
\begin{equation}
\label{eq:prf.1}
  \langle n+1,m' |S_k| n+1,m \rangle = \veps{i_1\ldots i_{n+1}}(m')^*
 (S_{(n+1)k})_{i_1\ldots i_{n+1};j_1\ldots j_{n+1}} \veps{j_1\ldots
   j_{n+1}}(m),
\end{equation}
with $-n-1\leq m',m \leq n+1$. From (\ref{eq:s.mat.prop.a}), the
matrix $\vec{S}_{(n+1)}$ can be written as
\begin{equation*}
(S_{(n+1)k})_{i_1\ldots i_{n+1};j_1\ldots j_{n+1}} =
(S_{(n)k})_{i_1\ldots i_{n};j_1\ldots j_{n}}  \delta_{i_{n+1}j_{n+1}} + 
\delta_{i_1j_1}\ldots \delta_{i_nj_n} (S_{(1)k})_{i_{n+1};j_{n+1}}.
\end{equation*}
Substituting this equality together with (\ref{eq:std.tens.s}) into
the right-hand side of (\ref{eq:prf.1}) leads to
  \begin{align*}
&\veps{i_1\ldots i_{n+1}}(m')^*
 (S_{(n+1)k})_{i_1\ldots i_{n+1};j_1\ldots j_{n+1}} \veps{j_1\ldots
   j_{n+1}}(m) =
\sum_{\mu',\mu=-n}^n \sum_{\nu',\nu=-1}^1
\CG{n}{\mu'}{1}{\nu'}{n+1}{m'} \CG{n}{\mu}{1}{\nu}{n+1}{m}\\
&\times \left(
\veps{i_1\ldots i_{n}}(\mu')^*(S_{(n)k})_{i_1\ldots i_{n};j_1\ldots
  j_{n}} \veps{j_1\ldots j_{n}}(\mu) \delta_{\nu'\nu} +
\delta_{\mu'\mu} \veps{i_{n+1}}(\nu')^*(S_{(1)k})_{i_{n+1};j_{n+1}}
\veps{j_{n+1}}(\nu) 
\rule{0pt}{14pt}\right)\\
&=
\sum_{\mu',\mu=-n}^n \sum_{\nu',\nu=-1}^1
\CG{n}{\mu'}{1}{\nu'}{n+1}{m'} \CG{n}{\mu}{1}{\nu}{n+1}{m}
\left(\rule{0pt}{14pt}
\langle n,\mu'|(S_A)_k|n,\mu \rangle \delta_{\nu'\nu} +
\delta_{\mu'\mu }\langle 1,\nu'|(S_B)_k|1,\nu \rangle
\rule{0pt}{14pt}\right),\\
\intertext{where in the last equality we used the case $n=1$ given by
  (\ref{eq:spin.mat.elm.1}) and the inductive hypothesis
  (\ref{eq:spin.matrix.element}) to express the matrix elements of
  $S_{(n)k}$ and $S_{(1)k}$ on the second line in terms of the matrix
  elements of a spin operator $\vec{S}_A$ for a spin-$n$ system $A$
  and a spin operator $\vec{S}_B$ for a spin-$1$ system $B$,
  respectively.  $\vec{S}_{A,B}$ being both angular-momentum
  operators, their matrix elements are given by (\ref{eq:cartj}). The
  total angular-momentum operator for the total system $A+B$ is
  $\vec{S}=\vec{S}_A\otimes I_B + I_A \otimes \vec{S}_B$, and we have}
&= \sum_{\mu',\mu=-n}^n \sum_{\nu',\nu=-1}^1
\CG{n}{\mu'}{1}{\nu'}{n+1}{m'} \CG{n}{\mu}{1}{\nu}{n+1}{m}
\langle n,\mu'; 1,\nu'|S_k|n,\mu, 1,\nu \rangle\\
&=
\langle n,1,n+1,m'|S_k|n,1,n+1,m \rangle
=\langle n+1,m'|S_k|n+1,m \rangle,
\hspace{77.9ex}\qed
  \end{align*}
where the last equality is (\ref{eq:prop.1}).

\setcounter{equation}{0}
\section{Non--maximally--coupled bipolar spherical harmonics} 
\label{sec:non-maximal.tshy}

The bipolar spherical harmonics \cite{var88} defined in
(\ref{eq:bipolar}) have been given an explicit tensorial expression in
\cite{bou2}.  Here we quote that expression for
$Y^{\ell\ell'}_{jm}(\verr{},\ver{r}\,')$ with a slightly different
notation that is more appropriate to the purposes of this paper. We
assume $\ell'\geq\ell$ for notational simplicity. The case
$\ell'<\ell$ results from the relation
$Y^{\ell\ell'}_{jm}(\verr{},\ver{r}\,')=(-1)^{\ell'+\ell-j}
Y^{\ell'\ell}_{jm}(\ver{r}\,',\verr{})$ that follows from the
definition (\ref{eq:bipolar}).  Thus, for $\ell'\geq\ell$ from
\cite{bou2} we find
\begin{subequations}
\label{eq:bipolar.gen.not}
\begin{equation}
  \label{eq:bipolar.general}
  \begin{aligned}
&Y^{\ell\ell'}_{jm}(\verr{},\ver{r}\,') =
\frac{i^\nu}{4\pi}\sqrt{2^j\nu!} \sqrt{\binom{2j}{j+n}} 
\sqrt{\frac{(2j+1)(2\ell+1)(2\ell'+1)}{(\ell'+\ell+j+1)!}}
\underset{\scriptstyle k_1+k_2=n}
{\displaystyle\sum_{k_1=k_{1\mathrm{min}}}^n 
\sum_{k_2=0}^{k_{2\mathrm{max}}}} (-1)^{k_2} \binom{n}{k_1}
\sum_{q=q_\mathrm{min}}^{q_\mathrm{max}}\\
&\times  (2q-1)!!
\binom{t}{2q} \veps{i_1\ldots i_{k_2+q}j_1\ldots j_{k_1+q}h_1\ldots
h_{t-2q}}(m) \verr{i_1}\ldots \verr{i_{k_2+q}} \ver{r}\,'_{j_1}
\ldots \ver{r}\,'_{j_{k_1+q}} v_{h_1}\ldots v_{h_{t-2q}}
P^{(j-q)}_{\ell+k_1}(\verr{}\cdot\ver{r}\,'),
  \end{aligned}
\end{equation}
where the following notations have been used
\begin{equation}
  \label{eq:bipolar.notation}
  \begin{gathered}
\nu=\ell'+\ell-j,
\quad
n=\ell'-\ell,
\quad
t=j-n,
\quad
\vec{v}=\verr{}\wedge\ver{r}\,',  
\quad
q_\mathrm{min} = \max\{0,\ell+k_2-\nu\},
\\
k_{1\mathrm{min}} =
\begin{cases}
  \max\{0,n-\frac{\nu}{2}\}\\
  \max\{0,n-\frac{\nu-1}{2}\}
\end{cases},
\quad
k_{2\mathrm{max}} =
\begin{cases}
  \min\{\frac{\nu}{2},n\}\\
  \min\{\frac{\nu-1}{2},n\}
\end{cases},
\quad
q_{\mathrm{max}} =
\begin{cases}
  \ell-\frac{\nu}{2} & \text{$\nu$ even}\\
  \ell-\frac{\nu-1}{2}-1 & \text{$\nu$ odd}
\end{cases},
  \end{gathered}
\end{equation}
\end{subequations}
and where $P^{(k)}_n(x)$ is the $k^\mathrm{th}$ derivative of the
Legendre polynomial $P_n(x)$.  In (\ref{eq:bipolar.general}) it is
understood that the factor $\verr{i_1}\ldots \verr{i_{k_2+q}}$ is to
be replaced by 1 if $k_2+q<1$, and analogously the other products of
$\ver{r}\,'$ and $\vec{v}$.  It is not difficult to check that
(\ref{eq:bipolar.gen.not}) reduces to (\ref{eq:max.bshy}) in the case
of maximal coupling $\nu=0$.  In the case $\nu>0$ of non-maximally
coupled $Y^{\ell\ell'}_{jm}$, (\ref{eq:bipolar.gen.not}) gives its
dual cartesian irreducible tensor just as (\ref{eq:max.bshy}) does for
maximally coupled ones.

\setcounter{equation}{0}
\section{Proof of (\ref{eq:partir3})} 
\label{sec:appD}

To derive (\ref{eq:partir3b}) we start from the second equality in
(\ref{eq:nder.8}) written in the form
\begin{equation}
  \label{eq:appD1}
  (Y^{\ell 1}_{\ell m}(\verr{}))_j = -\frac{i}{\sqrt{\ell(\ell+1)}}
  N_\ell \varepsilon_{jab} r_a \partial_b\left(
    \veps{(\ell)k_1\ldots k_\ell}(m) \verr{k_1}\cdots \verr{k_\ell}
\right),
\end{equation}
where we have used (\ref{eq:shyfinal}). By multiplying both sides of
(\ref{eq:appD1}) by $r^\ell$, substituting in it definition
(\ref{eq:tshy}) for $Y^{\ell 1}_{\ell m}$ and using the relation
$\varepsilon_{jab} r_a r^\ell \partial_b(f(\vec{r})) =
\varepsilon_{jab} r_a \partial_b(r^\ell f(\vec{r}))$ on its right-hand
side, we get
\begin{equation}
  \label{eq:appD2}
  \begin{aligned}
\sum_{\mu,\nu}\CG{\ell}{\mu}{1}{\nu}{\ell}{m}\veps{(\ell)i_1\ldots
  i_\ell}(\mu) \veps{(1)j}(\nu) r_{i_1}\ldots r_{i_\ell} &=
-\frac{i}{\sqrt{\ell(\ell+1)}} \varepsilon_{jab}r_a
\veps{(\ell)k_1\ldots k_{\ell}}(m) \partial_b(r_{k_1}\ldots r_{k_\ell})
\\
&=-i \sqrt{\frac{\ell}{\ell+1}}\varepsilon_{jab}r_a
\veps{(\ell)bk_1\ldots k_{\ell-1}}(m)r_{k_1}\ldots r_{k_{\ell-1}},
  \end{aligned}
\end{equation}
where in the second equality we made use of the total symmetry of
$\veps{(\ell)}$. Division of both sides of (\ref{eq:appD2}) by
$r^\ell$ yields its useful form
\begin{equation}
  \label{eq:appD3}
(Y^{\ell 1}_{\ell m}(\verr{}))_j =   -i
\sqrt{\frac{\ell}{\ell+1}}\varepsilon_{jih} \verr{i} \veps{(\ell)hj_1\ldots
  j_{\ell-1}}(m)\verr{j_1}\ldots \verr{j_{\ell-1}}. 
\end{equation}
Differentiation of both sides of (\ref{eq:appD2}) with
$\partial_{j_1}\ldots\partial_{j_\ell}$ leads to
\begin{equation}
  \label{eq:appD4}
\ell!\sum_{\mu,\nu}\CG{\ell}{\mu}{1}{\nu}{\ell}{m}\veps{(\ell)i_1\ldots 
  i_\ell}(\mu) \veps{(1)j}(\nu)  =
-i \sqrt{\frac{\ell}{\ell+1}} (\ell-1)! \sum_{h=1}^\ell
\varepsilon_{ji_hb}
\veps{(\ell)bi_1\ldots\widehat{i}_h\ldots i_{\ell}}(m),
\end{equation}
which is (\ref{eq:partir3b}).  We see that (\ref{eq:nder.8}),
(\ref{eq:partir3b}) and (\ref{eq:appD1})--(\ref{eq:appD3}) are
different ways of rewriting the same equality.

The derivation of (\ref{eq:partir3c}) starts with the third equality
in (\ref{eq:nder.8}) written in the form
\begin{equation}
  \label{eq:appD5}
  r^n (Y^{n 1}_{(n-1) m}(\verr{}))_i = \frac{1}{\sqrt{n(2n-1)}} \left(
\partial_i(r^{n+1}Y_{(n-1)m}(\verr{})) - (2n+1) r^{n-1}
Y_{(n-1)m}(\verr{}) r_i
\right).
\end{equation}
Differentiation of both sides of this equality with 
$\partial_{j_1}\ldots\partial_{j_n}$ and use of (\ref{eq:tshy}) and
(\ref{eq:tot.sym.alt.def}) leads to
\begin{equation}
  \label{eq:appD6}
  \begin{aligned}
  \sum_{\mu,\nu} \CG{n}{\mu}{1}{\nu}{n-1}{m} \veps{\{n,n\}j_1\ldots
    j_n}(\mu) \veps{i}(\nu) &= 
\frac{1}{\sqrt{n(2n-1)}}
\left(
\frac{\lambda'_{(n,n)}}{\lambda'_{(n+1,n-1)}}
\veps{\{n+1,n-1\}j_1\ldots j_ni}(m)\right.\\
&\quad\left.-(2n+1)\frac{\lambda'_{(n,n)}}{\lambda'_{(n-1,n-1)}}
\sum_{k=1}^n \veps{\{n-1,n-1\}j_1\ldots \widehat{j}_k\ldots j_n}(m) \delta_{j_ki}
\right),
  \end{aligned}
\end{equation}
where, as above, a caret over a subindex indicates that it is
omitted.  Substitution of (\ref{eq:tot.sym.def}) in (\ref{eq:appD6})
finally leads to
\begin{equation}
  \label{eq:appD7}
  \begin{aligned}
\lefteqn{
  \sum_{\mu,\nu} \CG{n}{\mu}{1}{\nu}{n-1}{m} \veps{(n)j_1\ldots
    j_n}(\mu) \veps{i}(\nu) = } \hspace{1.25ex}   \\
&=
\frac{1}{\sqrt{n(2n-1)}}\left(
\frac{\lambda'_{(n,n)}}{\lambda'_{(n+1,n-1)}}\lambda_{(n+1,n-1)}\frac{1}{(n+1)!} 
\veps{(n-1)\{j_1\ldots j_{n-1}}(m)\delta_{j_ni\}}
-(2n+1) \frac{\lambda'_{(n,n)}}{\lambda'_{(n-1,n-1)}}\right.\\
&\quad\left.\times
\sum_{k=1}^n \veps{(n-1)j_1\ldots \widehat{j}_k\ldots j_n}(m)
\delta_{j_ki} \right)\\
&=
\frac{1}{n\sqrt{(2n-1)(2n+1)}}\left(
\frac{1}{(n-1)!} 
\veps{(n-1)\{j_1\ldots j_{n-1}}(m)\delta_{j_ni\}}
-(2n+1) 
\sum_{k=1}^n \veps{(n-1)j_1\ldots \widehat{j}_k\ldots j_n}(m)
\delta_{j_ki} \right)\\
&=
\frac{1}{n\sqrt{(2n-1)(2n+1)}}\left( 2
\sum_{1\leq k<h\leq n}\veps{(n-1)j_1\ldots \widehat{j}_k
\ldots \widehat{j}_h\ldots j_{n}i}(m)\delta_{j_kj_h}
-(2n-1) 
\sum_{k=1}^n \veps{(n-1)j_1\ldots \widehat{j}_k\ldots j_n}(m)
\delta_{j_ki} \right),
  \end{aligned}
\end{equation}
the last equality being (\ref{eq:partir3c}).

\end{document}